\documentclass[a4paper,11pt]{amsart}
\usepackage{latexsym}
\usepackage{amssymb}
\usepackage{amsmath}
\usepackage{amsthm}
\usepackage{amsfonts}
\usepackage{amssymb}
\usepackage{booktabs}
\usepackage{graphicx}
\usepackage{epic, xypic}
\usepackage{float}
\restylefloat{table}

\numberwithin{equation}{section}

\newcommand{\be}{\begin{equation}}
\newcommand{\ee}{\end{equation}}
\newcommand{\bes}{\begin{equation*}}
\newcommand{\ees}{\end{equation*}}
\newcommand{\eqn}{\begin{eqnarray}}
\newcommand{\feqn}{\end{eqnarray}}
\newcommand{\eqnn}{\begin{eqnarray*}}
\newcommand{\feqnn}{\end{eqnarray*}}

\newtheorem{prop}{Proposition}

\begin{document}

\begin{titlepage}

     \thispagestyle{empty}
    \begin{flushright}
        \hfill {CERN-PH-TH/2014-026} \\

    \end{flushright}


    \begin{center}

         {\huge{\textbf{Iwasawa nilpotency degree of non compact symmetric cosets in $\mathbf{\mathcal{N}}\!-$extended Supergravity\\}}}
         \vspace{20pt}

{\Large{{\bf Sergio L. Cacciatori$^{1}$}, {\bf Bianca L. Cerchiai$^2$}, \\{\bf Sergio Ferrara$^{3}$}, {\bf Alessio Marrani$^{4}$}}}

        \vspace{5pt}

         {$^1$ Dipartimento di Scienze ed Alta Tecnologia,\\Universit\`{a} degli Studi dell'Insubria, Via Valleggio 11, 22100 Como, Italy\\
         and INFN, Sezione di Milano, Via Celoria 16, 20133 Milano, Italy\\
        \texttt{sergio.cacciatori@uninsubria.it}}

         \vspace{5pt}

         {$^2$ Dipartimento di Matematica,\\Universit\`{a} degli Studi di Milano, Via Saldini 50, 20133 Milano, Italy\\
         and INFN, Sezione di Milano, Via Celoria 16, 20133 Milano, Italy\\
        \texttt{bianca.letizia.cerchiai@cern.ch}}

         \vspace{5pt}

        {$^3$ Physics Department, Theory Unit, CERN,\\CH 1211, Geneva 23, Switzerland\\
and INFN - Laboratori Nazionali di Frascati,\\Via Enrico Fermi
40, I-00044 Frascati, Italy\\
and Department of Physics and Astronomy,\\ University of California, Los Angeles, CA 90095-1547, USA\\
        \texttt{sergio.ferrara@cern.ch}}

         \vspace{5pt}

{$^4$ Instituut voor Theoretische Fysica, KU Leuven,\\
Celestijnenlaan 200D, B-3001 Leuven, Belgium\\
\texttt{alessio.marrani@fys.kuleuven.be}}

 \vspace{15pt}

\end{center}



We analyze the polynomial part of the Iwasawa realization of the coset representative of non compact symmetric Riemannian spaces. We start by
studying the role of Kostant's principal $SU(2)_{P}$ subalgebra of simple Lie algebras, and how it determines the structure of the nilpotent
subalgebras. This allows us to compute the maximal degree of the polynomials for all faithful representations of Lie algebras. In particular the metric coefficients are related to the scalar kinetic terms while the representation of electric and magnetic charges is related to the coupling of scalars to vector field strengths as they appear in the Lagrangian.

We consider symmetric scalar manifolds in $\mathcal{N}\!-$extended supergravity in various space-time
dimensions, elucidating various relations with the underlying Jordan algebras and normed Hurwitz algebras. For magic supergravity theories, our
results are consistent with the Tits-Satake projection of symmetric spaces
and the nilpotency degree turns out to depend only on the space-time dimension of the theory.

These results should be helpful within a deeper investigation of the corresponding
supergravity theory, \textit{e.g.} in studying ultraviolet properties of maximal supergravity in various dimensions.


\end{titlepage}
\newpage

\tableofcontents


\section{\label{Intro}Introduction}

In the present paper we analyze the polynomial part of the Iwasawa realization of the coset representative of non compact symmetric Riemannian spaces.

Such non compact forms of the Lie groups are relevant for supergravity \cite{CFS, CJ, Julia, GST, CVP, dWVVP, Luciani, ADFT-1}. Here the scalar sector of the theory is described by a non linear sigma model based on a non compact symmetric space with negative curvature $G/K$, where $G$ is non compact and $K$ is its maximal compact subgroup. In this framework $G$ is the continuos group of $U$-duality transformations, called electric magnetic dualities in $D=4$.
Exceptions to this description are $\mathcal{N}=1$ and $2$ theories, where the sigma model doesn't need to be a symmetric space but it has a restricted holonomy (Hodge-K\"ahler for $\mathcal{N}=1$, special K\"ahler and Quaternionic for vector and hypermultiplets in $\mathcal{N}=2$).

Moreover, non compact homogeneous spaces and their non compact duality group $G$ appear in the description of the black hole orbits \cite{FerMal, FerGun, LPS, BFGM} as well as duality invariant Bekenstein-Hawking entropy formula \cite{KK, Stro, ADF, CFMZ-1} according to the attractor mechanism \cite{FKS, FK, ADFT-2}. The study of the attractor mechanism in string theory has been pioneered in \cite{OSV}.

In the past few years there has been a strong development in 4-dimensional $\mathcal{N} =8$ and $\mathcal{N} =4$ supergravity theories (SUGRA) because of their unexpected remarkable ultraviolet behaviour. It has been discovered that $d=4$ $\mathcal{N} =8$ SUGRA is finite in the ultraviolet up to four loops \cite{Bern-1}, while the pure $\mathcal{N}=4$ theory is finite up to 3 loops \cite{Bern-2}, which is not a priori guaranteed by supersymmetry, where a valid counterterm could in principle exist. $\mathcal{N}=4$ SUGRA, however, is already divergent at four loops~\cite{Bern-3}, and when coupled with matter already at one loop~\cite{Fischler, Bern-4}.

Perturbative finiteness of $d=4$ $\mathcal{N} =8$ SUGRA is possible only if its $E_{7(7)}$ symmetry is anomaly-free \cite{Marcus, BHN, Kallosh, Beisert}. Even when the symmetry is anomalous, as in pure $\mathcal{N}=4$ SUGRA, its Ward identities can restrict counterterms \cite{BN, BHS-1, BHS-2}. The latter can remain invariant under the Borel parabolic subgroup \cite{BHS-2} of the original
non compact duality group.  At the quantum level the group is expected to be broken to one of its discrete subgroups. An example of this phenomenon is the Dirac-Schwinger-Zwanziger quantization of black hole states. 

The Iwasawa parametrization explictly constructed at the group level in \cite{CCM} provides a realization of the coset based on a nilpotent subalgebra. As the coset representative directly enters the Lagrangian of the corresponding supergravity theory, it naturally delivers a choice of the fields for which they appear polynomially in the Lagrangian, making the calculations significantly more manageable. This is relevant for the study of the ultraviolet properties of maximal supergravities in various dimensions.

Our analysis of the structure of the nilpotent subalgebra is based on the work of
Kostant \cite{Kostant-1} who has introduced the concept of principal $SU(2)$ subalgebra (principal triple). He has shown how such principal $SU(2)$ characterizes the nilpotent subalgebras of maximal degree of the enveloping algebra in the adjoint representation.  We generalize his results for the adjoint to all the faithful linear representations.

The principal triple is related to the appearance of $W$-algebras, which are relevant as symmetries for integrable Toda systems (for a review see e.g. \cite{per}) and higher spins \cite{CFPT}.
In particular, the $W$-algebra related to the principal SL(2) is finite and Abelian and the corresponding Toda system is called Abelian.
In the framework of AdS higher spins, the choice of an embedding of $SL(2)$ determines the asymptotical symmetry on the boundary, hence fixing the theory.
It turns out \cite{CHL} that only the $W$-algebras constructed using the principal embedding could admit a unitary representation for large values of
the central charge. 

We have been able to determine the degree of the polynomials in the nilpotent part of the Iwasawa decomposition for the faithful representations of all the non compact symmetric spaces. In particular we have applied our results to the symmetric scalar manifolds in Maxwell-Einstein theories of (super)gravity in various space-time dimensions for various signatures. We also determine the degree of the polynomials occurring in the biinvariant metric of the coset. We find that in the case of magic $\mathcal{N}=2$ supergravity the degree is independent of the particular representation appearing in the corresponding magic square, but depends only on the space-time dimension. On the other hand, in the case of theories associated to split normed algebras, like e.g. maximal supergravity, we discover an intriguing connection with the $1$-form potential representations of the electric-magnetic ($U$-)duality group.
\vspace{1em}

The plan of the paper is as follows.

In Sec.~\ref{sec:nil} we introduce the general setup, and define the Iwasawa nilpotency degree in a given representation of a Lie group. Then, in Secs.~\ref{sec:split} and \ref{sec:non-split} the nilpotency degree is computed for the maximally non compact (\textit{split}) form and for any other real
form, respectively. Some examples are considered in Sec.~\ref{sec:rep-dep}. The degree of the polynomials occurring in the biinvariant metric of the
coset is then computed in Sec.~\ref{sec:metric}.

Sec.~\ref{sec:summary} presents three Tables summarizing the reasonings and the main results of this investigation.

As an application, in Sec.~\ref{sec:Supergravity} we consider the construction of the coset representative of non compact Riemannian symmetric
manifolds, namely of symmetric scalar manifolds occurring in Maxwell-Einstein theories of gravity in various Lorentzian space-dimensions,
possibly endowed with local supersymmetry. Our results are presented in a number of Tables. As respectively analyzed in Secs.~\ref{q-indep} and
\ref{q-dep}, for magic supergravity theories our results are consistent with the Tits-Satake projection of symmetric spaces, whereas for theories related to
split normed algebras (such as maximal supergravity), we find intriguing connections with the $1$-form potential representations of the
electric-magnetic ($U$-)duality group. In Sec.~\ref{sec:Universal} we also comment on the relation between our results and the axionic $U$-duality
generators related to five dimensions, exhibiting the universal degree of nilpotency $4$.

Five Appendices conclude the paper. In Apps.~\ref{App-Racah} and \ref{app:semispin}, some basic facts on the Racah-Casimir polynomials of a Lie
algebra $\mathfrak{g}$ and on the semispin groups are recalled. Then, in Apps.~\ref{app:Inverse-Cartan}, \ref{app:Dynkin-diagrams} and \ref{app:Satake-Type} the inverse Cartan matrices, the Dynkin diagrams and the Satake-type vectors of simple Lie algebras are reported.


\section{\label{sec:nil}Iwasawa Nilpotency Degrees in a given Representation}

Let us consider a \textit{simple} Lie group $G$ of rank $l$. The corresponding root lattice $\Lambda _{R}$ is generated by $l$ simple roots
$\alpha _{1},\ldots ,\alpha _{l}$, and it is a sub-lattice\footnote{In general, $\Lambda _{R}$ is a proper sub-lattice of $\Lambda _{W}$,
satisfying $\Lambda _{W}/\Lambda _{R}\simeq Z$, where $Z$ is the center of the covering group of $G$. For center-free groups, such as $G_{2}$, $F_{4}$,
$E_{8}$, the two lattices $\Lambda _{R}$ and $\Lambda _{W}$ do coincide.} of the weight lattice $\Lambda _{W}$, which is the integer lattice generated by
the fundamental weights $\mu ^{1},\ldots ,\mu ^{l}$. The simple roots define a real space $\mathbb{R}^{l}=\Lambda _{R}\otimes \mathbb{R}$, naturally
endowed with a positive definite scalar product $(|)$, inherited from the Killing form of the Lie algebra $\mathfrak{g}$ of $G$. The fundamental
weights are univocally related to the simple roots by
\begin{equation}
\langle \mu ^{i}|\alpha _{j}\rangle :=2\frac{(\mu ^{i}|\alpha _{j})}{(\alpha _{j}|\alpha _{j})}=\delta _{j}^{i}.
\end{equation}

The set of dominant weights is defined as the intersection between the weight lattice and the closure of the convex cone generated by the
fundamental weights
\begin{equation}
\Lambda _{W}^{+}:=\{m_{1}\mu^{1}+\ldots +m_{l}\mu^{l}~|~m_{i}\in \mathbb{N\cup }\left\{ 0\right\} ~\}.
\end{equation}
The dominant weights are in one-to-one correspondence with the irreducible representations (irreps.) of $\mathfrak{g}$, since each dominant weight is
the maximal weight of an (unique up to isomorphisms) irrep., and all irreps. are of this kind. In particular, all irreps. can be obtained by $G$
-covariantly branching the tensor products of the representations associated to the fundamental weights (thus explaining the name \textit{fundamental}
\footnote{Often, in the physical literature the name \textquotedblleft fundamental representation\textquotedblright\ is reserved to the non-trivial, smallest
(fundamental) irrep. For example, for the $\mathfrak{su}(N)$ algebras, all fundamental representations, and then all representations, can be obtained
from the external powers of the smallest one: $V(\mu^{k})=\wedge ^{k}V(\mu^{1})$.
\par
This can justify the physical notation. However this is not true for all cases. For example, for $\mathfrak{so}(2m-1)$ all but one fundamental
irreps. can be obtained from the smallest one, and all but two can be obtained for $\mathfrak{so}(2m)$. Indeed, the spinor representations must be
constructed in an independent way, and in this sense are not \textquotedblleft less fundamental" than the smallest one.}). The
representations associated to the fundamental weights $\mu _{i}$ are called
fundamental representations $V(\mu _{i})$ ($i=1,\ldots ,l$); they are reported in App. \ref{app:Dynkin-diagrams} for all simple Lie groups.

Now, let us consider an irrep. $V(\mu_{M})$ of $\mathfrak{g}$ associated to any maximal weight $\mu _{M}=m_{1}\mu^{1}+\ldots +m_{l}\mu^{l}$. In such a
representation, one can fix a basis of matrices $h_{i}$ for the (image under the representation map of a) Cartan subalgebra of $\mathfrak{g}$. To any
simple root $\alpha _{j}$ ($j=1,...,l$) it corresponds a matrix $\lambda_{\alpha _{j}}\in End(V(\mu_{M}))$, the so-called \textit{root} matrix,
which behaves as the raising operator by $\alpha _{j}$: if $v$ is a weight vector with weight $\mu$, then either $\lambda_{\alpha _{j}} v$ is zero or
it is a weight vector with weight $\mu+\alpha_{j}$.\newline It is clear from this, that each root matrix is nilpotent of some order, and
that, in particular, the linear span $\mathfrak{n}$ of positive root matrices (that are root matrices associated to positive roots) is made of nilpotent elements.

We are interested in determining the degree of nilpotency of the generic element in $\mathfrak{n}$ or in a suitable proper subspace, related to certain Iwasawa
symmetric constructions, in any given representation. Let us first discuss the case of the symmetric spaces associated to the split real forms.


\subsection{\label{sec:split}Iwasawa Polynomials for the Split Real Form}

The maximally non compact \ (\textit{split}) form of a simple Lie group $G$ (of rank $l$) is the unique real form having $l$ non compact Cartan
generators, spanning the Cartan subalgebra $\mathfrak{C}$. Its \textit{maximal compact subgroup} ($mcs$) $H$ has dimension \mbox{$h=($dim$(G)-l)/2$}, and
it is the smallest maximal subgroup symmetrically embedded in $G$.

Let us consider the Iwasawa construction \cite{Helgason} of the non compact, irreducible, Riemannian, (globally) \textit{symmetric}
space\footnote{For a review on irreducible, Riemannian, globally symmetric (IRGS) spaces in supergravity, see \textit{e.g.} \cite{LA08}.}
\begin{equation}
M:=\frac{G}{H},  \label{M}
\end{equation}
which has rank $\mathbf{r}=l$, and (real) dimension $($dim$(G)+l)/2$. It is here worth remarking that, even though in the split cases the rank $\mathbf{r}$ of the quotient is the same as the rank $l$ of the group, we will use $\mathbf{r}$ in place of $l$ wherever the
rank of the manifold is relevant. This is in order to avoid confusion, and
it will make the generalization beyond the split case clearer.

Let us work with any irreducible representation $V(\mu _{M})$. Let $c_{1},\ldots ,c_{\mathbf{r}}$ be a basis of non compact generators for the Cartan
subalgebra $\mathfrak{C}$, and let us fix any choice of \textit{positive} roots with respect to $\mathfrak{C}$. These are exactly $h$, denoted by
$\alpha _{a}$, $a=1,\ldots ,h$, with $\lambda _{\alpha _{a}}$ being the corresponding eigenmatrices (named root matrices in the treatment above).
Now, the coset representative $\mathcal{M}$ of the symmetric space $M$ (\ref{M}) can be parameterized in terms of the parameters $y_{1},\ldots ,y_{\mathbf{r}}$
(coordinates of the maximal non--compact Cartan subalgebra $\mathfrak{C}$) and $x_{1},\ldots ,x_{h}$ (coordinates pertaining to the $h$ nilpotent
eigenmatrices $\lambda _{\alpha _{a}}$'s) as follows\footnote{This (total) Iwasawa parametrization is usually named
\textit{\textquotedblleft standard"} parametrization; it is used \textit{e.g.} in \cite{Magic-Coset-Decomp}.}:
\begin{equation}
\mathcal{M}[\vec{y},\vec{x}]=\exp \left( \sum_{i=1}^{\mathbf{r}}y^{i}c_{i}\right) \exp \left( \sum_{a=1}^{h}x^{a}\lambda _{\alpha _{a}}\right) =\left(
\prod\limits_{i=1}^{\mathbf{r}}\exp \left( y^{i}c_{i}\right) \right) \exp \left(\sum_{a=1}^{h}x^{a}\lambda _{\alpha _{a}}\right) .  \label{SIC-M}
\end{equation}
The generators $c_{i}$ can be chosen so that $\exp (iy^{j}c_{j})$ is periodic. Note that the number $\mathbf{I}=h$ of Iwasawa (nilpotent)
generators of the NISS $M$ (\ref{M}) generally satisfies $\dim _{\mathbb{R}}(M)=\mathbf{r}+\mathbf{I}$ (\textit{cfr.} Table~\ref{tab:fourierFg3}), but in this
very case also satisfies $\mathbf{I}=(\dim _{\mathbb{R}}(G)-l)/2$. The first factor on the r.h.s. of (\ref{SIC-M}) is the product of $\mathbf{r}$ exponentials of
the $y_{i}$-coordinates, whereas, by virtue of the nilpotency properties of $\mathfrak{n}$ discussed above, the second factor is a polynomial in the $x^{a}$'s, of
which we want to determine the degree.\newline
In order to achieve this, let us consider the Taylor expansion:
\begin{equation}
\exp \left( \sum_{a=1}^{h}x^{a}\lambda _{\alpha _{a}}\right)=\sum_{n=0}^{\infty }\frac{1}{n!}\left( \sum_{a=1}^{h}x^{a}\lambda _{\alpha_{a}}\right) ^{n},
\end{equation}
in which the terms of degree $n$ are given by
\begin{equation}
M_{n}:=\frac{1}{n!}\sum_{d_{1}+\ldots +d_{h}=n,~d_{i}\geq 0}\frac{(x^{1})^{d_{1}}}{d_{1}!}\cdots \frac{(x^{h})^{d_{h}}}{d_{h}!}\sum_{\sigma
}\lambda _{\sigma _{1}}\cdots \lambda _{\sigma _{n}},  \label{poly}
\end{equation}
where $\sigma $ are the permutations of $\{1^{\times d_{1}},\ldots,h^{\times d_{h}}\}\in {\mathbb{N}}^{n}$.\newline
Our problem is to determine the largest $n$ such that $M_{n}\neq 0$, whereas $M_{n+1}=0$. The answer for the case of the adjoint representation can be
found in \cite{Kostant-1}, proposition 5.4:\newline
Let $\alpha _{1},\ldots ,\alpha _{l}$ be the simple roots. All other positive roots are expressed in terms of these as linear combinations with
non-negative integer coefficients. In particular, there exists a longest root, whose coefficient are all maximal:
\begin{equation}
\alpha _{L}=n_{1}\alpha _{1}+\ldots +n_{l}\alpha _{l}.  \label{longest}
\end{equation}
Set
\begin{equation}
\mathbf{q}:=n_{1}+\ldots +n_{l}.  \label{d_S}
\end{equation}
Then, the proposition of Kostant states that the maximal nilpotency of the elements of $\mathfrak{n}$ is $2\mathbf{q}+1$, or equivalently, the maximal polynomial
degree is $2\mathbf{q}$. In particular, this is reached iff all the coefficients of the simple roots are non vanishing.

Indeed, for all simple groups, $\mathbf{q}$ is given by
\begin{equation}
\mathbf{q}=j_{l}=C_{G}-1,
\end{equation}
where $C_{G}$ is the \textit{Coxeter number} of $G$, and $j_{l}$ is the \textit{maximal} spin of the $\mathfrak{su}(2)_{P}$-irreps. into which the
adjoint irrep. $\mathbf{Adj}$ of $G$ branches under the maximal embedding
\begin{eqnarray}
G &\supset &SU(2)_{P};  \label{pre-emb-1} \\
\mathbf{Adj}\left( G\right) &=&\sum_{A=1}^{l}\left( \mathbf{2j}_{A}+\mathbf{1}\right) =\sum_{A=1}^{l}\mathbf{S}_{j_{A}},~~~j_{1}=1<....<j_{l}.
\label{embb-1}
\end{eqnarray}
$\mathbf{S}_{j_{A}}$ denotes the $\mathfrak{su}(2)_{P}$-irrep. of spin $s=j_{A}$ of the unique \textit{principal} $SU(2)_{P}$ embedded in $G$
(\textit{cfr.} App. \ref{App-Racah}).\newline
Let us now show how the Kostant proposition works and how it can be generalized to any other representation. A generic element $x\in \mathfrak{n}$ will take the form
\begin{equation*}
x=x^{1}\lambda _{1}+\ldots +x^{h}\lambda _{h},
\end{equation*}
where $\lambda _{i}$, $i=1,\ldots ,h$ are all the positive root matrices. With \textit{generic} we mean that the coefficients corresponding to the
simple roots are all non vanishing. Our aim is to compute the maximal power $d$ such that $x^{d}\neq 0$ but $x^{d+1}=0.$ To do this we can compute the
matrix elements of $x^{d}$ w.r.t. some basis of the representation vector space $V$. A particularly convenient basis is the one given by a choice of
weight vectors. If $v\in V_{\mu }$ corresponds to the weight $\mu $, then
\begin{equation*}
xv\in V_{\mu +\alpha _{1}}+\ldots +V_{\mu +\alpha _{h}},
\end{equation*}
where $V_{\mu +\alpha _{i}}=0$ if $\mu +\alpha _{i}$ is not a weight of the given representation. It is then clear that by repeatedly acting with $x$ one gets
\begin{equation*}
x^{k}v\in \bigoplus_{i_{1},\ldots ,i_{k}}V_{\mu +\alpha _{i_{1}}+\ldots +\alpha _{i_{k}}}.
\end{equation*}
Since the representation is characterized by a highest weight $\mu _{M}$, for $k$ large enough, all $V_{\mu +\alpha _{i_{1}}+\ldots \alpha _{i_{k}}}$
will vanish. The matrix corresponding to $x^{k}$ vanishes if $x^{k}v=0$ for any weight vector $v$. Now, among all the weights there exists also a lowest
weight $\mu _{m}$, with the property that $\mu _{m}-\alpha _{i}$ is not a weight if $\alpha _{i}$ is a positive root, and the corresponding eigenspace
is one dimensional. Let $v_{m}$ be the corresponding eigenvector in the fixed basis. It follows immediately from the properties of the weight system
that if $x^{k}v_{m}=0$ then $x^{k}v=0$ for any other weight vector $v$. Thus we can work only on $v_{m}$ to check nilpotency. \newline
Now, the action of $x$ on $v_{m}$ is a combination of shifting of $v_{m}$ by positive roots (in the sense discussed above). So, we can look at the
repeated multiplication by $x$ as a combination of paths in the set of weights, starting by $v_{m}$ and adding a positive root at each step (in all
different possible ways). A path will stop when the addition of a further positive root will bring it out of the the set of weights. This is
equivalent to saying that all the paths will stop at the highest weight. Note that at a given step $x^{k}$, it may happen that one or more paths, but not
all, stop. We call $k$ the length of the path. The corresponding term will be given by\footnote{$v_{M}$ is the eigenvector in the basis corresponding
to the maximal weight $\mu _{M}$.} $v_{M}$ with coefficient a monomial of degree $k$ in the $x^{i}$. At the next step the corresponding term will be
sent to zero, but $x^{k+1}$ will not vanish. It follows that the degree $d$ we are looking for corresponds to the step $d$ for which all the surviving
paths (eventually only one) reach the highest weight. Thus, $d$ is the maximal possible length of a path from the lowest weight to the highest one.
\newline
Since a positive root is a combination of simple roots with integer non negative coefficients, any maximal path can be obtained by using only simple
roots to realize each step. Thus, we immediately arrive to the following conclusion: setting
\begin{equation}
\mu _{M}-\mu _{m}=\sum_{i=1}^{l}q_{i}\alpha _{i}, \label{maxmenmin}
\end{equation}
where $\alpha _{i}$, $i=1,\ldots ,l$ are the simple roots, and $q_{i}$ are positive integer coefficients, the length of a maximal path is
\begin{equation}
d=\sum_{i=1}^{l}q_{i}.
\end{equation}

\subsubsection{Computation of the Degree $d$}

We now state the following
\begin{prop}
Let $\Lambda$ a finite dimensional complex simple Lie algebra of rank $l$, and $\alpha_i$, $i=1,\ldots,l$ a choice of simple roots. Let $(\rho, V)$ the
irreducible representation of $\Lambda$ associated to the highest weight $\mu_M$. Set
\begin{eqnarray}
\mu_M=\sum_{i=1}^l r_i \alpha_i.  \label{maximal}
\end{eqnarray}
Then
\begin{eqnarray}
s=\sum_{i=1}^l r_i \in \mathbb{Z}/2.
\end{eqnarray}
Moreover, assume $x=\sum_{i=1}^l x^i \alpha_i$ to be generic. Then $\exp x$ is a polynomial of degree $d=2s$ in the $x^i$.
\end{prop}

\begin{proof}
Let $\mathfrak{C}$ the Cartan subalgebra. Since the simple roots form a basis of ${\mathfrak{C}}^*$, the decomposition (\ref{maximal}) is always
possible and unique, with $r_i$ non negative rational coefficients. Let $\mu_m\in {\mathfrak{C}}^*$ the minimal weight. We know that there exist
positive integer coefficients $q_i$ such that (\ref{maxmenmin}) holds.

Now, from the definition of $\mu_m$, there exist non negative coefficients $\tilde r_i$, $i=1,\ldots,l$ such that
\begin{eqnarray}
\mu_m=-\sum_{i=1}^l \tilde r_i \alpha_i.  \label{minimal}
\end{eqnarray}
Inserting this into (\ref{maxmenmin}), we see that $r_i+\tilde r_i=q_i$ is an integer and
\begin{eqnarray}
d=\sum_{i=1}^l (r_i+\tilde r_i).  \label{degree}
\end{eqnarray}
Since $\{-\alpha_1,\ldots,-\alpha_l\}$ is a good candidate for a fundamental root system, we know that there exists an element $w$ in the Weyl reflection
group and a permutation $\sigma\in {\mathcal{P}}_l$ such that $-\alpha_i=w(\alpha_{\sigma(i)})$. In particular $\mu_m=w(\mu_M)$. Indeed
$w(\mu_M)$ is a weight. Assuming it is not the minimal weight, it would exist at least a simple root $\alpha_i$ such that $w(\mu_M)-\alpha_i=w(\mu_M+
\alpha_{\sigma(i)})$ is also a weight. Then, $\mu_M+\alpha_{\sigma(i)}$ would be also a weight, thus contradicting the maximality of $\mu_M$.
\newline
Applying this to (\ref{minimal}) we get $\tilde r_i=r_{\sigma(i)}$. Then
\begin{eqnarray}
\sum_{i=1}^l \tilde r_i =\sum_{i=1}^l r_{\sigma(i)}= \sum_{i=1}^l r_i,
\end{eqnarray}
which from (\ref{degree}) implies
\begin{eqnarray}
d=2\sum_{i=1}^l r_i.
\end{eqnarray}
\end{proof}

Note that if $x$ was not generic, then some simple root would not enter in constructing the path and the degree of the polynomial would turn out to be
lower. We are going to discuss this point in the next subsection. Note also that we have assumed for the algebra to be complex. However the proposition
holds true also for the split form. Indeed, in the split case all the root matrices are in the real algebra and we do not need to complexify in order
to proceed with the proof.\newline
Using the proposition and the results in \cite{Kostant-1}, we see that $s$ is the spin of the $SL(2)_{P}$ principal subgroup of $\Lambda $ in the given
representation\footnote{More precisely $s$ is the highest spin representation in the direct decomposition of $(\rho ,V)$ under $SL(2)_{P}$.}.
In order to compute $d$ for the given representation, it is sufficient to determine the degrees $d^{i}=2s^{i}$ associated to the fundamental weights
$\mu ^{i}$. From the proof of the proposition it is clear that if $\mu _{M}=\sum_{i}m_{i}\mu ^{i}$ for non negative integers $m_{i}$, then
\begin{equation*}
s=\sum_{i=1}^{l}m_{i}s^{i}.
\end{equation*}
Note that the degree $d$ coincides with the {\textit{level}} of $\mu_{m}$, which is the number of steps necessary to reach $\mu _{m}$ by acting
on $\mu _{M}$ with the negative of the simple roots. These have been computed for all fundamental representations and simple groups in Table~10 of \cite{Slansky}.
It is however instructive and useful for the next application to see how they can be obtained.

\

Set
\begin{eqnarray}
\mu^i=\sum_{j=1}^l r^{ij} \alpha_j.
\end{eqnarray}
By definition of fundamental weights
\begin{eqnarray}
\delta^i_k=\langle \mu^i | \alpha_k \rangle=\sum_{j=1}^l r^{ij} \langle \alpha_j | \alpha_k \rangle= \sum_{j=1}^l r^{ij} C_{jk},
\end{eqnarray}
where $\pmb C=\{C_{jk} \}$ is the Cartan matrix. Thus, we see that
\begin{eqnarray}
r^{ij}=C^{ij}\equiv (\pmb C^{-1})^{ij},
\end{eqnarray}
are given by the inverse of the Cartan matrix.\newline
The principal spin associated to the representation $V(\mu^i)$ is then
\begin{eqnarray}
s^i=\sum_{j=1}^l C^{ij}.
\end{eqnarray}
A convenient way to express this result is as follows. Let $\bar s_\Lambda$ be the column vector whose entries are the principal spins $s^i$ of the
fundamental representations of $\Lambda$, and let $\bar \varepsilon$ the column vector in $\mathbb{R}^l$ with all entries equal to $1$. Then
\begin{eqnarray}
\bar s_\Lambda=\pmb C_{\Lambda}^{-1} \bar\varepsilon,
\end{eqnarray}
where we have specified the dependence of the Cartan matrix from $\Lambda$. $2\bar s_\Lambda=R_\Lambda^t$ is the transpose of the level vectors
$R_\Lambda $ given in Table~10 of \cite{Slansky}.


\subsection{\label{sec:non-split}Iwasawa Polynomials for the other non compact, Real Forms}

The construction we have developed up to now is valid for the Iwasawa decomposition associated to the split real form. For a generic real form not
all the simple roots enter in the Iwasawa parametrization, so that the expected degree for the Iwasawa polynomials is lower than the one previously
computed. We thus need to adapt our analysis to such general cases.

\

Therefore, let us consider any \textit{non-split}, real\ (non compact) form of $G$, where $G$ is a Lie group such that $Lie(G)=\Lambda$. In this case
\begin{equation}
\mathbf{r}:=\text{rank}(G/H)<\text{rank}(G)=l,
\end{equation}
so that the largest possible intersection $\mathbf{C}_{\mathfrak{p}}$ between the non compact Cartan subalgebra $\mathfrak{C}$ and the maximal
Cartan subalgebra $\mathbf{C}$ of $G$ has dimension $\mathbf{r}<l$. Thus, $\mathbf{C}$ can be chosen as follows:
\begin{equation}
\mathbf{C:}=\mathbf{C}_{\mathfrak{h}}\oplus \mathbf{C}_{\mathfrak{p}},\qquad \mathbf{C}_{\mathfrak{p}}:=\mathfrak{C}\cap \mathfrak{p},\qquad \text{dim}
\mathbf{C}_{\mathfrak{p}}=\mathbf{r},\quad \mathbf{C}_{\mathfrak{h}}\subset \mathfrak{h},~~~\text{dim}\mathbf{C}_{\mathfrak{h}}=l-\mathbf{r},
\end{equation}
where $\mathfrak{p}$ is the orthogonal complement (in the Lie algebra $\mathfrak{g}$ of $G$) of the maximal compact subalgebra $\mathfrak{h}$ of
$H=mcs\left( G\right)$; in other words, $\mathfrak{p}$ denotes nothing but the generators of the coset $G/H$ in the symmetric Cartan decomposition
\begin{equation}
\mathfrak{g}=\mathfrak{h}\oplus \mathfrak{p}.  \label{Cd-1}
\end{equation}
Note that generally $\mathbf{C}_{\mathfrak{h}}$ is not the maximal Cartan subalgebra of $\mathfrak{h}$.

Let us fix a basis $k_{1},\ldots ,k_{s}$ for $\mathbf{C}_{\mathfrak{h}}$, and $h_{1},\ldots ,h_{\mathbf{r}}$ for $\mathbf{C}_{\mathfrak{p}}$ ($s+\mathbf{r}=l$). Thus,
the Cartan decomposition (\ref{Cd-1}) can be further refined as follows:
\begin{equation}
\mathfrak{g}=\mathfrak{h}\oplus \left( \mathbf{C}_{\mathfrak{p}}\oplus \widetilde{\mathfrak{p}}\right) ,  \label{Cd-2}
\end{equation}
where $\widetilde{\mathfrak{p}}$ is the complement of $\mathbf{C}_{\mathfrak{p}}$ in $\mathfrak{p}$. It is here worth recalling that, as the Cartan
decomposition (\ref{Cd-1}) pertains to the maximal and \textit{symmetric} embedding of $\mathfrak{h}$ into $\mathfrak{g}$, it holds that
\begin{equation}
\lbrack \mathfrak{h},\mathfrak{h}]\subseteq \mathfrak{h},\qquad \lbrack \mathfrak{p},\mathfrak{p}]\subseteq \mathfrak{h},\qquad \lbrack \mathfrak{h},
\mathfrak{p}]\subseteq \mathfrak{p}.  \label{Cd-1-symm}
\end{equation}
Let $\mathfrak{k}$ be the Lie algebra of the \textit{normalizer} $\mathbf{K}$ of $\mathbf{C}_{\mathfrak{p}}$, \textit{i.e.} the largest Lie subalgebra of
$\mathfrak{h}$ such that $[\mathfrak{k},\mathbf{C}_{\mathfrak{p}}]=0$. Then (\ref{Cd-2}) and (\ref{Cd-1-symm}) imply that
\begin{equation}
\lbrack \mathfrak{k},\mathbf{C}_{\mathfrak{p}}]=0\Rightarrow \lbrack \widetilde{\mathfrak{h}},\mathbf{C}_{\mathfrak{p}}]\subseteq
\widetilde{\mathfrak{p}},\qquad \lbrack \widetilde{\mathfrak{p}},\mathbf{C}_{\mathfrak{p}}]\subseteq \widetilde{\mathfrak{h}},  \label{struttura}
\end{equation}
where $\widetilde{\mathfrak{h}}$ is the complement of $\mathfrak{k}$ in $\mathfrak{h}$.

It is here worth pointing out a convenient characterization of the roots of the \textit{simple}, non compact non split, real Lie algebra $\mathfrak{g}$.
Since $\mathbf{C}_{\mathfrak{h}}$ is contained in $\mathfrak{h}$ and commutes with $\mathbf{C}_{\mathfrak{p}}$, it is contained in $\mathfrak{k}$
and is necessarily a Cartan subalgebra for it. Then, it follows that
\begin{equation}
\text{rank}(\mathfrak{k})=\text{dim}(\mathbf{C}_\mathfrak{h})=s.
\end{equation}
One can then represent all roots of $\mathfrak{g}$ as the simultaneous eigenvalues of the operators \linebreak
$(ad_{k_{1}},\ldots,ad_{k_{s}};ad_{h_{1}},\ldots ,ad_{h_{\mathbf{r}}})$. By defining $k:=$dim$\mathbf{K}$, the eigenvectors of the roots $\alpha _{\mathfrak{h},a}$
($a=1,\ldots ,k-s$) of $\mathfrak{k}$ are in the complexification of $\mathfrak{k}$ and thus in the kernel of $ad_{h_{i}}$ ($i=1,\ldots ,\mathbf{r}$) : as
a consequence, the last $\mathbf{r}$ components of the string of simultaneous eigenvalues of $(ad_{k_{1}},\ldots ,ad_{k_{s}};ad_{h_{1}},\ldots,ad_{h_{\mathbf{r}}})
$ are zero. Actually, the roots $\alpha _{\mathfrak{h},a}$ are all the non-vanishing roots of $\mathfrak{g}$ with this property: all the other ones
are characterized by (\textit{at least} some) non-vanishing components among
the last $r$ ones in the string of simultaneous eigenvalues of $(ad_{k_{1}},\ldots ,ad_{k_{s}};ad_{h_{1}},\ldots ,ad_{h_{\mathbf{r}}})$. We will call
the corresponding roots $\alpha _{\mathfrak{p},b}$ ($b=1,\ldots ,2(h-k)$).
Indeed, these correspond to the non-vanishing roots of the $h_{i}$, whose number, by (\ref{struttura}), is $h-k$.\newline
Thus, the roots can be separated into the $\alpha_{\mathfrak{p}}$ and $\alpha_{\mathfrak{h}}$ roots. In realizing the quotient only the first ones
survive, the last ones being projected to zero roots. As usual, one can divide the roots $Rad$ in positive ones $Rad^{+}$ and negative ones $Rad^{-}$, namely:
\begin{equation}
Rad=Rad^{+}\oplus Rad^{-}.
\end{equation}
Correspondingly, this determines a decomposition of the root system associated to the Lie algebra $\mathfrak{p}$ of the symmetric coset $G/H$:
\begin{equation}
Rad_{\mathfrak{p}}=Rad_{\mathfrak{p}}^{+}\oplus Rad_{\mathfrak{p}}^{-}.
\end{equation}
While in the split case $Rad_{\mathfrak{p}}$ is a reduced lattice system, in the non-split case this is generally not true anymore, and generically each
root $\alpha $ is characterized by a multiplicity $m_{\alpha }\geq 1$; moreover, given a root $\alpha $, it can happen that $2\alpha $ or
$\frac{1}{2}\alpha $ is also a root.\newline

All the root systems $\Lambda _{G/H}$ associated to irreducible symmetric spaces $G/H$ have been classified by Araki \cite{mac}. The split forms are
exactly the ones which satisfy $\Lambda _{G/H}=R_{G}$, where $R_{G}$ is the root system of $G$ itself. For a generic non compact real form of $G$,
$\Lambda_{G/H}$ is not reduced, but it contains a maximal subsystem of roots (its \textit{reduction}), which correspond to a root system of simple kind
(but in general with non-trivial multiplicities). The complement of the reduction of a non-reduced $\Lambda _{G/H}$ possibly consists of further
roots which are twice the reduced ones (they are dubbed \textit{double }roots).\newline
\smallskip

Now we can extend our results to the given real form. In this case the nilpotent elements are of the form $x=\sum_{i} x^i \lambda_i$, where the sum
is extended over all root matrices corresponding to the positive roots of the lattice associated to the quotient symmetric manifold. Thus, not all the
simple roots matrices will participate to the Iwasawa construction, but only the ones corresponding to roots of $Lie(G)$ that are not sent to zero in the
projection on the quotient. As a consequence, in order to compute the degree of the corresponding Iwasawa polynomials, we have to project (\ref{maxmenmin})
on the quotient.\footnote{Note that a maximal weight is again maximal in the non reduced lattice after projection, but it is in general not unique.}
If we call $\pi$ the projection, we get
\begin{eqnarray}
\pi(\mu_M) -\pi(\mu_m)=\sum_{i=1}^\mathbf{r} q_{j_i} \pi(\alpha_{j_i}),
\end{eqnarray}
where $\alpha_{j_i}$ are the roots that are not in the kernel of $\pi$. We can then proceed exactly as in the proof of the previous proposition to get
for the degree of the Iwasawa polynomial
\begin{eqnarray}
d=2\sum_{i=1}^\mathbf{r} r_{j_i}.
\end{eqnarray}
Again, it is sufficient to compute only the spins associated to the fundamental representations, which as a consequence are
\begin{eqnarray}
d^i=\sum_{k=1}^\mathbf{r} C^{ij_k}.
\end{eqnarray}
Now, it is easy to determine which roots are projected into zero and which not. Indeed, to each real form, or, equivalently, to each non compact
Irreducible Symmetric Space (NISS), it is associated a Satake diagram, indicating as black dots the simple roots in the kernel of $\pi$. A complete
list of Satake diagrams can be found in \cite{mac}. Let us call $x$ of generic type in the Iwasawa nilpotent in the representation $(\rho,V)$ if
\begin{eqnarray}
x=\sum x^i \lambda_i,
\end{eqnarray}
where the sum is extended to all positive quotient roots, and the coefficients corresponding to simple roots are all non vanishing. Thus, we
have proved the following proposition:

\begin{prop}
Let $\Lambda$ be a complex simple Lie algebra of rank $l$ and $\Lambda^T$ be its real form corresponding to a given symmetric space of type $T$ and rank
$\mathbf{r}$ (as listed in \cite{CDPS}). Consider the associated Satake diagram and let $\bar \varepsilon^T$ the column vector in $\mathbb{R}^l$ with entry $1$
if corresponding to a white dot and zero otherwise (\textit{Satake vector}), and let $\{\bar e_i\}_{i=1}^l$ the canonical basis of $\mathbb{R}^l$. Let $x$
be of generic type in the corresponding Iwasawa nilpotent in the representation $V(\mu^i)$, ($i=1,\ldots,l$). Then
\begin{eqnarray}
P(x)=\exp (x)
\end{eqnarray}
is a polynomial of degree $d^i_{\Lambda^T}=2s^i_{\Lambda^T}$, where
\begin{eqnarray}
s^i_{\Lambda^T}=\bar e_i \cdot \pmb C^{-1} \bar \varepsilon^T.
\end{eqnarray}
\end{prop}

It should be pointed out that the same construction can be extended to reducible representations and to semi-simple groups. In these cases the maximal
polynomial degree will be determined by the maximal spin among all sub representations or among the simple factors, respectively. 

A list of the inverse Cartan matrices can be found in appendix \ref{app:cartan}, whereas all the Satake vectors are listed in appendix~\ref{app:satake}.


\subsubsection{\label{sec:rep-dep}Some Examples}

In order to illustrate how the degree of nilpotency depends on the $G$-rep\-re\-sen\-tation of the coset representative $\mathcal{M}$, we consider the
simplest case provided by the rank-$1$ symmetric space
\begin{equation}
\frac{G}{H}=\frac{SL(2,\mathbb{R})}{SO(2)},  \label{ex-1}
\end{equation}
which corresponds to setting $n=1$ in the first row of Tables~\ref{tab:fourierFg2} and \ref{tab:degrees}. The result $d=1$ pertains to the coset representative in the fundamental irrep. $\mathbf{2}$ (spin $s=1/2$) of
$G_{nc}=SL(2,\mathbb{R})$, with a $2\times 2$ coset representative $\mathcal{M}$; this case is relevant both to four-dimensional $\mathcal{N}=4$ \textit{\textquotedblleft pure"} supergravity as well as to $\mathcal{N}=2$
supergravity \textit{minimally coupled} to one Abelian vector multiplet (this latter theory can indeed be obtained from the former by a consistent
truncation of graviphotons and gravitinos), and it has been recently reconsidered in Sec. 5 of \cite{CFGM-d}.

On the other hand, the coset (\ref{ex-1}) can also be considered as the scalar manifold of the so-called $T^{3}$ model of $\mathcal{N}=2$, $D=4$
Maxwell-Einstein supergravity, in which the unique Abelian vector multiplet is coupled non-minimally, but rather through a cubic holomorphic
prepotential $\mathcal{F}=T^{3}$, to the gravity multiplet. In this case, the relevant irrep. of $G_{nc}=SL(2,\mathbb{R})$ is the $\mathbf{4}$ (spin
$s=3/2$), and thus the coset representative $\mathcal{M}$ is a $4\times 4$ matrix. Since the irrep. $\mathbf{4}$ has weight $3\lambda _{1}$, from the
reasoning above one obtains that the corresponding maximal degree is three times the fundamental one, and thus
\begin{equation}
\text{in~}\mathbf{4}\text{ of~}SL(2,\mathbb{R}):d=3.  \label{spin-3/2}
\end{equation}
The result (\ref{spin-3/2}) on the degree of nilpotency $d+1=4$ of the unique Iwasawa generator is consistent with the treatment recently given in
\cite{CFGM-d}; indeed, for the $T^{3}$ model, which uplifts to five-dimensional minimal \textit{\textquotedblleft pure"} supergravity, the
(partial) axionic Iwasawa construction exploited in \cite{CFGM-d} (in particular, \textit{cfr.} Sec. 2 and App. B of \cite{CFGM-d}) is actually a
total Iwasawa construction.

For any $\mathfrak{sl}(2,\mathbb{R})$-irrep. of weight $m\lambda_1$ (and thus spin $s=m/2$), the corresponding maximal Iwasawa degree concerning the
IRGS (\ref{ex-1}) reads $d_{S}=2s=m$, and thus the corresponding Iwasawa nilpotency degree is nothing but $m+1$.

The negative constant scalar curvature of (\ref{ex-1}) for the $s=1/2$ and $s=3/2$ cases respectively is $R=-2$ and $R=-2/3$ (corresponding to $m=1$ and
$m=3$), and they are the unique values for which this symmetric K\"{a}hler coset is a \textit{special} manifold~\cite{CVP}.\medskip

Considering IRGS relevant as scalar manifolds in supergravity theories (see \textit{e.g.} \cite{LA08} for a comprehensive review), one should consider
only a few other \textit{non-fundamental} irreps., such as the rank-$3$ antisymmetric skew-traceless $\mathbf{14}^{\prime }$ of
$Sp\left( 6,\mathbb{R}\right) $.

Regarding the coset manifold $SL(2,\mathbb{R}) \over SO(2)$, another remark is in
order. As it is discussed in the next Section~\ref{sec:metric}, a further quantity that is interesting is the invariant metric tensor. Then it can be seen immediately that for $SU(1,1) \over U(1)$ the degree $d_g$ of the metric vanishes: $d_g=0$. This is a consequence of the fact that in this case the metric does not depend on the axions and hence on the nilpotent part of the Iwasawa decomposition. In other words, the metric of the upper half plane depends only on the imaginary part Im$(z)$ of the complex coordinate $z=$Re$(z)+i$ Im$(z)$. For these results, compare the values $d_g=0$ in the first row for $n=1$ and in the third row of Table~\ref{tab:SKG-D=4}.


\section{\label{sec:metric}The Degree $d_{g}$ of Metric Polynomials}

There is a second problem we are interested in. Let us consider a NISS $M$. It can be realized in a given representation $(\rho, V)$ by means of the
Iwasawa parametrization, which, following our previous notations, can be formally written as
\begin{eqnarray}
M=\exp (y) \exp (x),
\end{eqnarray}
where $y=\sum_{i=1}^\mathbf{r} y^i c_i$, where $c_i$, $i=1,\ldots,\mathbf{r}$ is a basis of $\mathbf{C}_{\mathfrak{p}}$, and $x=\sum_i x^i \lambda_i$ is of generic type
in the Iwasawa nilpotent $\mathfrak{n}$ in the representation $(\rho,V)$, apart from positive codimension submanifolds. Note that
$\mathbf{C}_{\mathfrak{p}}\oplus \mathfrak{n}$ is a (non simple) subalgebra so that the Lie algebra valued right-invariant one form
\begin{eqnarray}
J_M:= dM M^{-1}
\end{eqnarray}
again lies in $\mathbf{C}_{\mathfrak{p}}\oplus \mathfrak{n}$. From this, one obtains the invariant metric tensor $g_{ij}$ on $M$ in the standard way. In
particular, $J_M$ will be polynomial in $x$ of a given degree $d$, so that the metric will be polynomial of degree $2d$. The aim of this section is to
determine $d$.\newline
To this end, let us first note that,
\begin{eqnarray}
J_M= d \exp(y)\ \exp(-y) +Ad_{\exp (y)} (d \exp(x)\ \exp(-x)).
\end{eqnarray}
Since $\tilde J(x):=d \exp(x)\ \exp(-x)$ lies in $\mathfrak{n}$, on which $Ad_{\exp (y)}$ acts diagonally, the polynomial part comes only from $\tilde J(x)$. Now
\begin{eqnarray}
\tilde J(x)= \sum_i dx^i \int_0^1 Ad_{\exp (vx)}(\lambda_i) dv,
\end{eqnarray}
from which we can easily read the degree of $\tilde J(x)$. Indeed, for a fixed $i$, $Ad_{\exp (vx)}=\exp (v ad_x)$ acts on $\lambda_i$ in the same
way discussed in the previous section: $\lambda_i$ is the root matrix corresponding to the positive root $\alpha_i$; the action of $ad_x$ changes
$\lambda_i$ to a combination of root matrices corresponding to roots of the form $\alpha_i+\alpha_j$, with $\alpha_j$ positive, and so on until reaching
the highest root. Thus we can follow exactly the same reasoning as in the previous section: the main point is that it does not matter what
representation we are starting from, as it is always the adjoint representation which is acting in this case. What is relevant is only the type of the NISS. For
example, if we work with split form, then we have to work with the highest root, expressed in terms of the simple roots as
\begin{eqnarray}
\alpha_H=\sum_{i=1}^l n^i \alpha_i.
\end{eqnarray}
Since we construct the paths by starting from a positive root $\lambda_i$, the longest paths will be obtained by adding simple roots to a starting
simple root. The maximal possible length of a path will then be $d=\sum_{i=1}^l n^i-1=C_G-2$. \newline
If we work with a generic real form, then, as before, we will need to sum only over the coefficients $n_{j_i}$ corresponding to the roots that are not in
the kernel of the projector $\pi$: $d=\sum_{i=1}^r n_{j_i}$. We thus arrive at the following result:

\begin{prop}
Let $M=G/H$ a NISS of type $T$ and $\mathbf{g}$ the corresponding positive definite biinvariant metric induced by the Killing form. Let $x$ the
coordinates in the nilpotent part of the standard Iwasawa parametrization of $M$. Then $\mathbf{g}$ is a polynomial function of $x$ of degree
\begin{equation*}
d_{g}=2(s_{ad}-1),
\end{equation*}
where $s_{ad}$ is the spin of the principal $SL(2)_{P}$ in the adjoint representation. If $\bar{n}=(n_{1},\ldots ,n_{l})$ are the coordinates of
the highest roots w.r.t. the simple roots, then
\begin{equation*}
d_{g}=2(\bar{n}\cdot \bar{\varepsilon}^{T}-1).
\end{equation*}
\end{prop}
Let us notice that, again, this result is easily generalized to the case of a reducible representation or to semisimple groups: the degree is the maximum degree among the irreducible or simple factors, respectively.


\section{\label{sec:summary}Summary of Results}

The results of the above reasonings are collected in Tables~\ref{tab:fourierFg2}, \ref{tab:fourierFg3} and \ref{tab:degrees}.

In the first column of Table~\ref{tab:fourierFg2}, consistent with the notation \textit{e.g.} in \cite{mac, CDPS}, we indicate the label for
the non compact, real form $G_{nc}$ of $G$, which in turn is given in the second column. $H=mcs(G_{nc})$ is given, along with possible discrete
factors, in the third column. Then, following the classification of \cite{mac}, the fourth column specifies the type of the root systems $\Lambda_{G/H}$
associated to irreducible symmetric space $G/H$. After \cite{CDPS}, the fifth column provides the coefficients $\left( n_{1},...,n_{\mathbf{r}}\right) $
of the longest root $\alpha _{L}$~(\ref{longest}). In the sixth column, the multiplicities for the roots and the double roots are given \cite{mac}: these
are just a number for simply-laced root systems; on the other hand, in non-simply-laced systems they are given as a pair, in which the first and
second entry give respectively the multiplicity of the long and short root.
\newline

In Table~\ref{tab:fourierFg3} we report the rank $\mathbf{r}$ of the NISS $G_{nc}/mcs\left( G_{nc}\right) $, along with the number $\mathbf{I}$ of the
corresponding nilpotent (\textit{i.e.}, Iwasawa) generators; the resulting (real) dimension $\dim $ of the NISS is nothing but $\mathbf{I}+\mathbf{r}$, where $\mathbf{r}$ is the rank of $G_{nc}$.; finally, in the fifth column the \textit{character} \cite{Helgason} $\chi :=nc-c$ of the coset is reported, where
\textquotedblleft $nc$" and \textquotedblleft $c$" denote respectively the number of non compact resp. compact generators of the NISS itself.

In Table~\ref{tab:degrees} we report the degrees of the Iwasawa polynomials associated to each NISS. In the first column we report $G_{nc}$, in the
second column the rank $l$ of $G_{nc}$, in column three we give the vector of degrees of Iwasawa polynomials associated to the fundamental
representations, and in the last column the degree $d_{g}$ of the polynomial part in the biinvariant metric $\mathbf{g}$.\newline
Note that from the fundamental degrees reported in the second column of Table~\ref{tab:degrees}, the nilpotency degree of the Iwasawa polynomials
can be obtained by simply adding one to each entry.

\begin{table}[tbph]
\begin{center}
\resizebox*{1\textwidth}{!}{\
\begin{tabular}{cccccc}
\toprule T & $G_{nc}$ & $H$ & $\Lambda _{G/H}$ & $\left(n_{1},...,n_{\mathbf{r}}\right) $ & $\vec{m}_{\lambda },\,\vec{m}_{2\lambda }$ \\
\midrule AI & $SL(n+1,\mathbb{R})$ & $SO(n+1)$ & $A_{n}\,(n\geq 1)$ & (1,1,\ldots ,1) & (1), (0) \\[0.2cm]
AII & $SU^{\ast }(2k)$ & $USp(2k)$ & $A_{k-1}\,(k>1)$ & (1,1,\ldots ,1) & (4), (0) \\[0.15cm]
$\mathrm{{AIII}_{a}}$ & $SU(p,q)$ & $S(U(p)\times U(q))$ & $B_{p}\,(1<p<q)$ & (2,2,\ldots ,2) & $2(1,q-p),(0,1)$ \\[0.15cm]
$\mathrm{{AIII}_{b}}$ & $SU(p,p)$ & $S(U(p)\times U(p))$ & $C_{p}\,(p>1)$ & $(2,2,\ldots ,2,1)$ & (1,2), (0,0) \\[0.15cm]
AIV & $SU(1,n)$ & $S(U(1)\times U(n))$ & $A_{1}$ & (2) & (2n-2), (1) \\[0.15cm]
$\mathrm{{BI_{a}}}$ & $SO(n,n+1)$ & $SO(n)\times SO(n+1)$ & $B_{n}\,(n\geq2) $ & (1,2,\ldots ,2) & (1,1),\thinspace\ (0,0) \\[0.15cm]
$\mathrm{{BI_{b}}}$ & $SO(p,q)$ & $SO(p)\times SO(q)$ & $B_{p}\,(1<p<n)$ & (1,2,\ldots ,2) & (1,2(n-p)+1), (0,0) \\[0.15cm]
BII & $SO(1,2n)$ & $SO(2n)$ & $A_{1}$ & (1) & (2n-1), (0) \\[0.15cm]
CI & $Sp(2n,\mathbb{R})$ & $U(n)$ & $C_{n}\,(n\geq 3)$ & (2,2,\ldots ,2,1) & (1,1),\thinspace\ (0,0) \\[0.15cm]
$CII_{a}$ & $USp(2p,2q)$ & $USp(2p)\times USp(2q)$ & $B_{p}\,(1\leq p\leq (n-1)/2)$ & (2,2,\ldots ,2) & (4,4n-8p), (0,3) \\[0.15cm]
$CII_{b}$ & $USp(2k,2k)$ & $USp(2k)\times USp(2k)$ & $C_{k}$ & $(2,2,\ldots,2,1)$ & (3,4), (0,0) \\[0.15cm]
$DI_{a}$ & $SO(n,n)$ & $SO(n)\times SO(n)$ & $D_{n}\,(n>3)$ & (1,2,\ldots,2,1,1) & (1), (0) \\[0.15cm]
$DI_{b}$ & $SO(n-1,n+1)$ & $SO(n-1)\times SO(n+1)$ & $B_{n-1}\,(n>2)$ & (1,2,\ldots,2) & (1,2), (0,0) \\[0.15cm]
$DI_{c}$ & $SO(p,q)$ & $SO(p)\times SO(q)$ & $B_{p}\,(1<p<n-1)$ & (1,2,\ldots,2) & (1,2(n-p)), (0,0) \\[0.15cm]
$DII$ & $SO(1,2n-1)$ & $SO(2n-1)$ & $A_{1}$ & (1) & (2n-2), (0) \\[0.15cm]
$DIII_{a}$ & $SO^{\ast }(4k+2)$ & $U(2k+1)$ & $B_{k}\,(k\geq 2)$ & (2,2,\ldots ,2) & (4,4), (0,1) \\[0.15cm]
$DIII_{b}$ & $SO^{\ast }(4k)$ & $U(2k)$ & $C_{k}\,(k\geq 2)$ & (2,2,\ldots,2,1) & (1,4), (0,0) \\[0.15cm]
G & $G_{2(2)}$ & $SO(4)/\mathbb{Z}_{2}$ & $G_{2}$ & (3,2) & (1,1), (0,0) \\[0.15cm]
FI & $F_{4(4)}$ & $USp(6)\times USp(2)$ & $F_{4}$ & (2,3,4,2) & (1,1), (0,0) \\[0.15cm]
FII & $F_{4(-20)}$ & $SO(9)$ & $A_{1}$ & (2) & (8), (7) \\[0.15cm]
EI & $E_{6(6)}$ & $USp(8)/\mathbb{Z}_{2}$ & $E_{6}$ & (1,2,2,3,2,1) & (1),(0) \\[0.15cm]
EII & $E_{6(2)}$ & $(USp(2)\times SU(6))/\mathbb{Z}_{2}$ & $F_{4}$ & (2,3,4,2) & (1,2), (0,0) \\[0.15cm]
EIII & $E_{6(-14)}$ & $(U(1)\times SO(10))/\mathbb{Z}_{4}$ & $B_{2}$ & (2,2) & (6,8), (0,1) \\[0.15cm]
EIV & $E_{6(-26)}$ & $F_{4}$ & $A_{2}$ & (1,1) & (8), (0) \\[0.15cm]
EV & $E_{7(7)}$ & $SU(8)/\mathbb{Z}_{2}$ & $E_{7}$ & (2,2,3,4,3,2,1) & (1),(0) \\[0.15cm]
EVI & $E_{7(-5)}$ & $(SU(2)\times SO(12))/\mathbb{Z}_{2}$ & $F_{4}$ & (2,3,4,2) & (1,4), (0,0) \\[0.15cm]
EVII & $E_{7(-25)}$ & $(U(1)\times E_{6})/\mathbb{Z}_{3}$ & $C_{3}$ & (2,2,1) & (1,8), (0,0) \\[0.15cm]
EVIII & $E_{8(8)}$ & $S_{s}(16)$ & $E_{8}$ & (2,3,4,6,5,4,3,2) & (1), (0) \\[0.15cm]
EIX & $E_{8(-24)}$ & $(SU(2)\times E_{7})/\mathbb{Z}_{2}$ & $F_{4}$ & (2,3,4,2) & (1,8), (0,0) \\
\bottomrule &  &  &  &  &
\end{tabular}
}
\end{center}
\caption{In this Table, we list the main ingredients necessary to describe the Iwasawa construction of the non compact,
irreducible, Riemannian, globally symmetric space $G_{nc}/H$. $S_{s}\left(16\right) $ denotes the semispin group of type $D_{8}$ \protect\cite{YokYas}
(see also appendix \protect\ref{app:semispin}). Note that, when listing the \textit{mcs} $H$ in the second column, the universal coverings are
considered, namely $SO(3)\equiv $ $Spin(3)$ $\simeq $ $SU(2)$, $USp(4)\simeq$ $SO(5)$ $\equiv $ $Spin(5)$, $USp(2)\simeq $ $SU(2)$, $SO(6)$ $\equiv $ $%
Spin(6)\simeq $ $SU(4)$, and $SO(4)$ $\equiv $ $Spin(4)\simeq $ $SU(2)\times SU(2)$. Furthermore, local isomorphisms among non compact, real forms are
used throughout (\textit{cfr.} \textit{e.g.} \protect\cite{Helgason}).
Notice that in $BI_b$, $p+q=2n+1$, in $CII_a$, $p+q=n$, and in $DI_c$, $p+q=2n$.}
\label{tab:fourierFg2}
\end{table}

\begin{table}[tbph]
\begin{center}
\resizebox*{1\textwidth}{!}{\
\begin{tabular}{cccccccc}
\toprule$G_{nc}$ & $\mathbf{r}$ &  & $\mathbf{I}$ &  & $\dim\left(G_{nc}/mcs\left( G_{nc}\right)\right) =\mathbf{I}+\mathbf{r}$ &  & $\chi $ \\
\midrule$SL(n+1,\mathbb{R})$ & $n$ &  & $\frac{\left( n+1\right) \left(n+2\right) }{2}-n-1$ &  & $\frac{\left( n+1\right) \left( n+2\right) }{2}-1$
&  & $n$ \\[0.2cm]
$SU^{\ast }(2k)$ & $k-1$ &  & $2k\left( k-1\right) $ &  & $k\left(2k-1\right) -1$ &  & $-2k-1$ \\[0.15cm]
$SU(p,q)$ & $\min \left( p,q\right) $ &  & $2pq-\min \left( p,q\right) $ & & $2pq$ &  & $-\left( p-q\right) ^{2}+1$ \\[0.15cm]
$SU(p,p)$ & $p$ &  & $2p^{2}-p$ &  & $2p^{2}$ &  & $1$ \\[0.15cm]
$SU(1,n)$ & $1$ &  & $2n-1$ &  & $2n$ &  & $-n\left( n-2\right) $ \\[0.15cm]
$SO(n,n+1)$ & $n$ &  & $n^{2}$ &  & $n\left( n+1\right) $ &  & $n$ \\[0.15cm]
$SO(p,q)_{p+q=2n+1}$ & $\min \left( p,q\right) $ &  & $pq-\min \left( p,q\right) $ &  & $pq$ &  & $[\left( p+q\right) -\left( p-q\right) ^{2}]/2$ \\[0.15cm]
$SO(1,2n)$ & $1$ &  & $2n-1$ &  & $2n$ &  & $n\left( 3-2n\right) $ \\[0.15cm]
$Sp(2n,\mathbb{R})$ & $n$ &  & $n^{2}$ &  & $n(n+1)$ &  & $n$ \\[0.15cm]
$USp(2p,2q)$ & $\min \left( p,q\right) $ &  & $4pq-\min \left( p,q\right) $ &  & $4pq$ &  & $-[\left( 2p+2q\right) +\left( 2p-2q\right) ^{2}]/2$ \\[0.15cm]
$USp(2k,2k)$ & $k$ &  & $k\left( 4k-1\right) $ &  & $4k^{2}$ &  & $-2k$ \\[0.15cm]
$SO(n,n)$ & $n$ &  & $n(n-1)$ &  & $n^{2}$ &  & $n$ \\[0.15cm]
$SO(n-1,n+1)$ & $n-1$ &  & $n\left( n-1\right) $ &  & $n^{2}-1$ &  & $n-2$ \\[0.15cm]
$SO(p,q)_{p+q=2n}$ & $\min \left( p,q\right) $ &  & $pq-\min \left( p,q\right) $ &  & $pq$ &  & $[\left( p+q\right) -\left( p-q\right) ^{2}]/2$ \\[0.15cm]
$SO(1,2n-1)$ & $1$ &  & $2\left( n-1\right) $ &  & $2n-1$ &  & $5n-2n^{2}-2$\\[0.15cm]
$SO^{\ast }(4k+2)$ & $2k+1$ &  & $\left( 4k+2\right) \left( k-1/2\right) $ & & $2k\left( 2k+1\right) $ &  & $-2k-1$ \\[0.15cm]
$SO^{\ast }(4k)$ & $2k$ &  & $4k\left( k-1\right) $ &  & $2k\left(2k-1\right) $ &  & $-2k$ \\[0.15cm]
$G_{2(2)}$ & $2$ &  & $6$ &  & $8$ &  & $2$ \\[0.15cm]
$F_{4(4)}$ & $4$ &  & $24$ &  & $28$ &  & $4$ \\[0.15cm]
$F_{4(-20)}$ & $1$ &  & $15$ &  & $16$ &  & $-20$ \\[0.15cm]
$E_{6(6)}$ & $6$ &  & $36$ &  & $42$ &  & $6$ \\[0.15cm]
$E_{6(2)}$ & $4$ &  & $36$ &  & $40$ &  & $2$ \\[0.15cm]
$E_{6(-14)}$ & $2$ &  & $30$ &  & $32$ &  & $-14$ \\[0.15cm]
$E_{6(-26)}$ & $2$ &  & $24$ &  & $26$ &  & $-26$ \\[0.15cm]
$E_{7(7)}$ & $7$ &  & $63$ &  & $70$ &  & $7$ \\[0.15cm]
$E_{7(-5)}$ & $4$ &  & $60$ &  & $64$ &  & $-5$ \\[0.15cm]
$E_{7(-25)}$ & $3$ &  & $51$ &  & $54$ &  & $-25$ \\[0.15cm]
$E_{8(8)}$ & $8$ &  & $120$ &  & $128$ &  & $8$ \\[0.15cm]
$E_{8(-24)}$ & $4$ &  & $108$ &  & $112$ &  & $-24$ \\
\bottomrule &  &  &  &  &  &  &
\end{tabular}
}
\end{center}
\caption{In the second column the rank $\mathbf{r}$ of the symmetric coset $G_{nc}/mcs\left( G_{nc}\right) $ is indicated, in the third column the
number $\mathbf{I}$ of Iwasawa (nilpotent) generators, in the fourth column the real dimension $\dim $ of $G_{nc}/mcs\left( G_{nc}\right) $, and in the
fifth column the character $\protect\chi :=nc-c$ of $G_{nc}$, where \textquotedblleft $nc$" and \textquotedblleft $c$" denote respectively the number of non compact and of compact generators of $G_{nc}$ itself \protect\cite{Gilmore}. Note that for the maximally non compact
(\textit{split}) forms, $\mathbf{r}=\protect\chi =l$, where we recall that $l$ is the rank of $G$ itself; consequently, for the split forms $\mathbf{I}$ is equal
to the number of positive (or of negative) roots of $G$. Also, we notice the peculiarity of the minimally non compact real form of $E_{6}$, i.e. of
$E_{6(-26)}$, for which $\protect\chi =-\dim\left(G_{nc}/mcs\left( G_{nc}\right)\right)$.}\label{tab:fourierFg3}
\end{table}
\begin{table}[tbph]
\begin{center}
\resizebox*{1\textwidth}{!}{\
\begin{tabular}{cccccc}
\toprule $G_{nc}$ & $l$ &  & $(d_1, \ldots, d_l)$ &  & $d_g$ \\\midrule $SL(n+1,\mathbb{R})$ & $n>0$ &  & $\left(\{i(n+1-i)\}_{i=1}^n\right) $ &  & $2(n-1)$ \\[0.2cm]
$SU^{\ast }(2k)$ & $2k-1,k>1$ &  & $(\{ 2(k-i)(i-1)+k-1, 2i(k-i)\}_{i=1}^{k-1}, k-1)$ &  & $2k-4$ \\[0.15cm]
$SU(p,q)_{1<p<q}$ & $p+q-1=n$ &  & $(\{i(2p+1-i)\}_{i=1}^{p-1},\{p(p+1)\}_{i=p}^q, \{(2n-i)(i+p-q+1)\}_{i=q+1}^{n})$ &  & $4p-2$ \\[0.15cm]
$SU(p,p)$ & $2p-1,p>1$ &  & $(\{ i(2p-i) \}_{i=1}^{2p-1})$ &  & $4p-4$ \\[0.15cm]
$SU(1,n)$ & $n>1$ &  & (2,\ldots,2) &  & 2 \\[0.15cm]
$SO(n,n+1)$ & $n>1$ &  & $(\{i(2n+1-i)\}_{i=1}^{n-1}, n(n+1)/2)$ &  & $4n-4$ \\[0.15cm]
$SO(p,q)_{p+q=2n+1,1<p<q}$ & $n=(p+q-1)/2$ &  & $(\{ i(2p+1-i) \}_{i=1}^p, \{ p(p+1) \}_{i=p+1}^n)$ &  & $4p-4$ \\[0.15cm]
$SO(1,2n)$ & $n>1$ &  & $(2,\ldots,2,1)$ &  & 0 \\[0.15cm]
$Sp(2n,\mathbb{R})$ & $n>1$ &  & $(\{i(2n-i)\}_{i=1}^n)$ &  & $4n-4$ \\[0.15cm]
$USp(2p,2q)_{0<p<q}$ & $p+q=n=2k$ &  & $(\{2(2i-1)(p+1)-2i^2, 2i(2p-i+1)\}_{i=1}^p,\{2p(p+1)\}_{i=p+1}^n)$ &  & $4p-2$ \\[0.15cm]
$USp(2k,2k)$ & $2k,k>0$ &  & $(\{2(2i-1)(2k-1)-2i^2+2i-1,2i(2k-i)\}_{i=1}^{k})$ &  & $4k-4$ \\[0.15cm]
$SO(n,n)$ & $n>2$ &  & $(\{i(2n-1-i)\}_{i=1}^{n-2},n(n-1)/2,n(n-1)/2)$ &  & $4n-8$ \\[0.15cm]
$SO(n-1,n+1)$ & $n>2$ &  & $(\{i(2n-1-i)\}_{i=1}^{n-2},n(n-1)/2,n(n-1)/2)$ & & $4n-8$ \\[0.15cm]
$SO(p,q)_{p+q=2n,1<p<q}$ & $n=(p+q)/2$ &  & $(\{2pi-i^2+i\}_{i=1}^{p-1}, \{p^2+p\}_{i=p}^n)$ &  & $4p-4$ \\[0.15cm]
$SO(1,2n-1)$ & $n>1$ &  & (2,\ldots,2,1,1) &  & 0 \\[0.15cm]
$SO^{\ast }(4k+2)$ & $2k+1,k>0$ &  & $(\{2ki-2\lfloor \frac{i}{2} \rfloor^2\}_{i=1}^{2k-1}, k^2+k-1, k^2+k-1)$ &  & $4k-2$ \\[0.15cm]
$SO^{\ast }(4k)$ & $2k,k>1$ &  & $(\{2ki-2\lfloor \frac{i}{2}\rfloor^2-i\}_{i=1}^{2k-2}, k^2-1, k^2)$ &  & $4k-4$ \\[0.15cm]
$G_{2(2)}$ & $2$ &  & $(10,6)$ &  & 8 \\[0.15cm]
$F_{4(4)}$ & $4$ &  & $(22, 42, 30, 16)$ &  & 20 \\[0.15cm]
$F_{4(-20)}$ & $4$ &  & $(4,8,6,4)$ &  & 2 \\[0.15cm]
$E_{6(6)}$ & $6$ &  & $(16, 30, 42, 30, 16, 22)$ &  & 20 \\[0.15cm]
$E_{6(2)}$ & $6$ &  & $(16, 30, 42, 30, 16, 22)$ &  & 20 \\[0.15cm]
$E_{6(-14)}$ & $6$ &  & $(6, 10, 14, 10, 6, 8)$ &  & 6 \\[0.15cm]
$E_{6(-26)}$ & $6$ &  & $(4, 6, 8, 6, 4, 4)$ &  & 2 \\[0.15cm]
$E_{7(7)}$ & $7$ &  & $(34, 66, 96, 75, 52, 27, 49)$ &  & 32 \\[0.15cm]
$E_{7(-5)}$ & $7$ &  & $(22, 42, 60, 46, 32, 16, 30)$ &  & 20 \\[0.15cm]
$E_{7(-25)}$ & $7$ &  & $(10, 18, 26, 21, 16, 9, 13)$ &  & 8 \\[0.15cm]
$E_{8(8)}$ & $8$ &  & $(92, 182, 270, 220, 168, 114, 58, 136)$ &  & 56 \\[0.15cm]
$E_{8(-24)}$ & $8$ &  & $(32, 62, 92, 76, 60, 42, 22, 46)$ &  & 20 \\
\bottomrule &  &  &  &  &
\end{tabular}
}
\end{center}
\caption{In the second column the rank $l$ of the group $G$ is indicated, in the third column the vector of fundamental degrees, that are the degrees of
the polynomials in the fundamental representations, and in the last column the degree of the polynomial part in the biinvariant metric $\mathbf{g}$. Recall that the nilpotency degree is given by $d+1$.
Note that generally the smallest irrep. yields the smallest degree of the corresponding polynomials. All computations are based on the inverse Cartan
matrices and Satake vectors listed in appendices \protect\ref{app:cartan} and \protect\ref{app:satake}. Note that the string
$\left(d_{1},...d_{l}\right) $ is the same for the spaces $DI_{a}$ and $DI_{b}$, as well as for $EI$ and $EII$; as given by the treatment of \protect\cite{mac},
this is due to the fact that only white nodes in the Tits-Satake diagram do contribute to the degree $d$, and in these non compact, real forms there are
no blackened nodes at all.\newline
Remark: Notice that for $SO^*(4k)$, following the convention in~\cite{mac}, we have assumed the even roots to be the non compact ones. However, from the symmetry of the Dynkin diagram, one sees that it is possible to obtain an isomorphic coset manifold by picking the last but one root instead of the last one as the non compact root.  In this case the last two entries of the third column need to be switched accordingly.}
\label{tab:degrees}
\end{table}

\newpage

\section{\label{sec:Supergravity}Application to Supergravity}

As an application of the set of general results for the degrees $d$ and $d_{g}$ of the Iwasawa polynomials resulting in the standard
construction of the coset representative of non compact Riemannian symmetric manifolds, we consider the \textit{symmetric} scalar manifolds of
supergravity theories in various Lorentzian space-dimensions (for a comprehensive review, see \textit{e.g.} \cite{LA08}), namely $D=3,4,5,6$. We
report the corresponding results in various Tables, and we address to the corresponding captions for further comments.

Summarizing, in Table~\ref{tab:SKG-D=4} we consider the special K\"{a}hler vector multiplets' scalar manifolds of $\mathcal{N}=2$ theories in $D=4$,
whereas in Table~\ref{tab:HKG-D=3} we deal with the quaternionic K\"{a}hler scalar manifolds~\cite{FS} (obtained through the so-called $c$-map \cite{CFG}) of the
theory dimensionally reduced (along a spacelike direction) down to $D=3$ (after dualization)~\cite{dWVVP}. Moreover, the vector multiplets' real special manifolds
(obtained through the so-called $R$-map \cite{GMRZ}) of the minimal theory obtained by oxidizing (\textit{i.e.}, uplifting) to $D=5$ Lorentzian dimension are
reported in Table~\ref{tab:RSG-D=5}. The scalar manifolds of \mbox{$\mathcal{N}>2$--extended} supergravity theories in $D=4$ and $D=5$ are considered in Tables~\ref{tab:D=4} and \ref{tab:D=5}, respectively, while the symmetric manifolds of $\mathcal{N}>4$--extended supergravity theories in $D=3$ are given in Table~\ref{tab:D=3}.

Concerning $D=6$, for brevity's sake we will consider only minimal chiral theories. The scalar manifold of $(1,0)$ \textit{chiral} magic supergravity
theories (based on $J_{2}^{\mathbb{A}}\sim \Gamma _{1,q+1}$, the Clifford algebra of $O\left( 1,q+1\right) $ \cite{JWVN}) is given by
$\frac{SO(1,q+1)}{SO\left( q+1\right) }$ (with $q=1,2,4,8$ for $\mathbb{A}=\mathbb{R},\mathbb{C},\mathbb{H},\mathbb{O}$, respectively), with the relevant
$G$-irrep. being the vector: $\mathbf{R}(SO(1,q+1))=\mathbf{q+2}=(1,0,...,0)$ (for a discussion of $D=6$ minimal chiral theories, see \textit{e.g.}
\cite{AFMT-1} and \cite{Gunaydin-Samtleben-D=6}). This formally matches the scalar manifold of the so-called
\textit{non-Jordan symmetric sequence} \cite{VPdW-1} (\textit{cfr.} line 1 of Table~\ref{tab:RSG-D=5}), and thus $d=2$ and $d_{g}=0$, regardless of the
value of $q$.

The general structures of (maximal and symmetric) group embeddings read as follows:
\begin{eqnarray}
QConf\left( J_{3}^{\mathbb{A}}\right) &\supset &Conf\left( J_{3}^{\mathbb{A}}\right) \times SL(2,\mathbb{R});  \label{prim} \\
Conf\left( J_{3}^{\mathbb{A}}\right) &\supset &Str_{0}\left( J_{3}^{\mathbb{A}}\right) \times SO(1,1);  \label{sec} \\
Str_{0}\left( J_{3}^{\mathbb{A}}\right) &\supset &Str_{0}\left( J_{2}^{\mathbb{A}}\right) \times SO(1,1),  \label{ter}
\end{eqnarray}
where $QConf$, $Conf$ and $Str_{0}$ respectively denote the quasi-conformal, conformal and reduced structure symmetry groups of Jordan algebras (see
\textit{e.g.} \cite{G-Lects}). In (\ref{prim}) $SL(2,\mathbb{R})$ is the so-called Ehlers group, related to the reduction of pure
Einstein gravity to $3$ dimensions, whereas in (\ref{sec}) and (\ref{ter}) $SO(1,1)$ is the Kaluza-Klein compactification factor. Due to the properties
of the Clifford algebra in certain dimensions, $Str_{0}\left( J_{2}^{\mathbb{A}}\right) $, in some cases contains a commuting non-trivial compact factor
$\frac{Tri\left( \mathbb{A}\right) }{SO\left( \mathbb{A}\right) }$, where $Tri\left( \mathbb{A}\right) $ and $SO\left( \mathbb{A}\right) $ denote the
triality group and norm-preserving group of the normed division algebra $\mathbb{A}$ (\textit{cfr. e.g.} \cite{Baez}). Note that
\begin{equation}
\frac{Tri\left( \mathbb{A}\right) }{SO\left( \mathbb{A}\right) }=\emptyset,SO(2),SO(3),\emptyset \text{~for~}\mathbb{A}=\mathbb{R},\mathbb{C},
\mathbb{H},\mathbb{O},
\end{equation}
and
\begin{equation}
Str_{0}\left( J_{2}^{\mathbb{A}}\right) =SO\left( 1,\dim _{\mathbb{R}} \mathbb{A}+1\right) \times \frac{Tri\left( \mathbb{A}\right) }{SO\left( \mathbb{A}\right) }.
\end{equation}
However, due to its compactness, this extra factor do not affect the structure of the (tensor multiplets') real scalar manifold in $D=6$.

By denoting with $G_{D,\mathbb{A}}$ the electric-magnetic ($U$-)duality group of the $J^{\mathbb{A}}$-based magic supergravity in $D$ Lorentzian
space-time dimensions, the supergravity interpretation of such symmetry groups is the following:
\begin{eqnarray}
QConf\left( J_{3}^{\mathbb{A}}\right) &=&G_{3,\mathbb{A}}; \\
Conf\left( J_{3}^{\mathbb{A}}\right) &=&G_{4,\mathbb{A}}; \\
Str_{0}\left( J_{3}^{\mathbb{A}}\right) &=&G_{5,\mathbb{A}}; \\
Str_{0}\left( J_{2}^{\mathbb{A}}\right) &=&G_{6,\mathbb{A}}.
\end{eqnarray}

In the second row of Table~\ref{tab:RSG-D=5} an extrapolation to $n=1$, which is not reported in the Table, corresponds to the $D=5$ uplift of the
so-called $ST^{2}$ model of $\mathcal{N}=2$, $D=4$ Maxwell-Einstein
supergravity. This is (the unique model) given by the dimensional
reduction of \textquotedblleft pure\textquotedblright\ $(1,0)$ minimal
chiral supergravity in $D=6$. For this model the maximal degrees vanish: $d=0$ and $d_g=0$.

In the second row of Table~\ref{tab:D=4} an extrapolation to $n=0$, which is not reported in the Table, corresponds to \textquotedblleft pure\textquotedblright\ $\mathcal{N}=4$, $D=4$ supergravity; the $U$-duality
group is $SL(2,\mathbb{R})\times SU(4)$, and the relevant $U$-irrep. is $(%
\mathbf{2},\mathbf{6})$. Indeed, the $U(1)$ factor of the $\mathcal{N}=4$ $%
\mathcal{R}$-symmetry $U(4)$ is gauged by the complex (\textit{axio-dilatonic%
}) scalar field of the gravity multplet, whereas the $SU(4)$ factor is not
gauged, \textit{i.e.} it is global. Concerning the non compact part of the $U
$-duality, this theory is like the $\mathcal{N}=2$ axio-dilatonic model
(based on $\overline{\mathbb{CP}}^{1}$, describing the \textit{minimal
coupling} of $1$ vector multiplet to the gravity multiplet \cite{Luciani};
\textit{cfr.} $n=1$ in the first row of Table~\ref{tab:SKG-D=4}), but in the $\mathcal{N}=4$ case the compact part is given by $SO(6)$ \textit{vs.} the $U(1)$ factor in the $\mathcal{N}=2$ model. It follows that the corresponding maximal Iwasawa degrees are $d=1$ and $d_g=0$.

In the first row of Table~\ref{tab:D=5} an extrapolation to $n=1$, which is not reported in the Table, corresponds to \textquotedblleft pure\textquotedblright\ $\mathcal{N}=4$, $D=5$ supergravity; the $U$-duality group is $SO(1,1)\times SO(5)$, where $SO(5)$ is the (global) $\mathcal{N}=4$ $\mathcal{R}$-symmetry. There is only one real scalar in the gravity multiplet, and the relevant $U$-rep. is $(\mathbf{1},\mathbf{1})+(\mathbf{1}, \mathbf{5})$ (\textit{cfr. e.g.} \cite{CFMZ-2}). The maximal Iwasawa degrees are $d=0$ and $d_g=0$.

\subsection{\label{q-indep}$q$-Independence of $d$ and $d_{g}$ for Magic Supergravities}

From Tables~\ref{tab:SKG-D=4}, \ref{tab:HKG-D=3}, \ref{tab:RSG-D=5} and the above observation implying $d=2$ and $d_g=0$ for $D=6$, one can report the values of $d$ and $d_{g}$
for magic Maxwell-Einstein supergravities in $D=3,4,5,6$ Lorentzian space-time dimensions in Table~\ref{tab:summary-1}.

Interestingly, for the magic Maxwell-Einstein supergravities \cite{GST}, the degrees $d$ and $d_{g}$ of the Iwasawa polynomial only
depend\footnote{Note that Table~\ref{tab:summary-1} considers all possible dimensions in which theories with $8$ local supersymmetries can be defined, namely $D=3,4,5,6$.} on $D$, and not on $q$; this seems consistent with the fact that the corresponding homogeneous (symmetric) manifolds are in the same Tits-Satake universality class (see \textit{e.g.} \cite{TS} for a general treatment).

Note that the second, third and fourth rows of Table~\ref{tab:summary-1} are the same as the ones of G\"{u}naydin-Sierra-Townsend (GST) Magic Square
\cite{GST} $\mathcal{L}_{3}\left( \mathbb{A}_{s},\mathbb{B}\right) $ (for a comprehensive treatment of $4\times 4$ Magic Squares $\mathcal{L}_{3}$'s,
\textit{cfr.} \cite{MS-1}). On the other hand, the relation between the first rows of Table~\ref{tab:summary-1} and of the GST Magic Square
$\mathcal{L}_{3}\left( \mathbb{A}_{s},\mathbb{B}\right) $ can be expressed by observing that the maximal compact subgroup ($mcs$) of the symmetry groups
occurring in the first row of Table~\ref{tab:summary-1} (namely, the $D=6$ $U $-duality group $G_{6,\mathbb{A}}=Str_{0}\left( J_{2}^{\mathbb{A}}\right)$)
is a (maximal symmetric) subgroup of the $mcs$ of $Str_{0}\left( J_{3}^{\mathbb{A}}\right) $, which is the group occurring in the first row of GST
Magic Square $\mathcal{L}_{3}\left( \mathbb{A}_{s},\mathbb{B}\right) $:
\begin{equation}
mcs\left( Str_{0}\left( J_{2}^{\mathbb{A}}\right) \right) \subset _{s}^{\max }mcs\left( Str_{0}\left( J_{3}^{\mathbb{A}}\right) \right).
\end{equation}
Such a relation can be trivially expressed by stating that the stabilizer of $D=6$ (tensor multiplets') real scalar manifold is maximally (and
symmetrically) embedded into the (vector multiplets') real special scalar manifold of the corresponding theory obtained by dimensional reduction
(along a spacelike direction) down to $D=5$ (for some discussion on anomaly-freedom conditions of minimal chiral theories in $D=6$, cfr. \textit{e.g.}
\cite{AFMT-1}).

\subsection{\label{q-dep}$q$-Dependence of $d$ and $d_{g}$ for $J^{\mathbb{A}_{s}}$-based Maxwell-Einstein (Super)Gravity Theories and $U$-Representations
of $1$-Form Potentials}

For theories based on $J^{\mathbb{A}_{s}}$ (where $\mathbb{A}_{s}$ denotes the \textit{split} form of the normed algebra $\mathbb{A}$), the values of
$d$ and $d_{g}$ are not $q$-independent, but rather they follow an interesting pattern. Essentially, the Iwasawa degree $d$ of the scalar manifold
$G_{D,\mathbb{O}_{s}}/mcs(G_{D,\mathbb{O}_{s}})$ in $D$ dimensions (in the relevant $G_{D,\mathbb{O}_{s}}$-irrep.) is related to the (real) dimension
of the $G_{D+1,\mathbb{O}_{s}}$-irrep. in which the $1$-form (Abelian) potentials of maximal supergravity sit (\textit{cfr.} \textit{e.g.}
\cite{Kleinschmidt,Brane Orbits}).

While the purely $D$-dependent values of $d$ and $d_{g}$ for $J^{\mathbb{A}}$-based magic Maxwell-Einstein supergravities are consistent with the
Tits-Satake projection for homogeneous (symmetric) manifolds \cite{TS}, it is here worth observing that the non compact symmetry ($U$-duality) groups
pertaining to $J^{\mathbb{A}_{s}}$-based theories are the maximally non compact (namely, \textit{split}) forms of Lie groups. As a consequence the corresponding homogeneous spaces are not necessarily in the same Tits-Satake universality class\footnote{We are grateful to Mario Trigiante for useful discussions on this point.}.

Tables~\ref{tab:DD=3}-\ref{tab:DD=5} respectively illustrate the resulting pattern for $D=3,4,5$ in some detail.

For the maximally supersymmetric theory based on $\mathbb{O}_{s}$, the uplifts to $D>5$ Lorentzian dimensions are known and well defined.

In $D=6$, one can consider the following embedding in the maximal non-chiral $\left( 2,2\right) $ supergravity:
\begin{eqnarray}
SO(5,5) &\supset &_{s}^{\max }SL(5,\mathbb{R})\times SO(1,1);  \label{emb-1}
\\
\mathbf{16} &=&\mathbf{10}_{-1}+\mathbf{5}_{3}+\mathbf{1}_{-5},
\label{emb-2}
\end{eqnarray}
and the Iwasawa degree $d$ of the $D=6$ scalar manifold $G_{6,\mathbb{O}_{s}}/mcs\left( G_{6,\mathbb{O}_{s}}\right) =SO(5,5)/(SO(5)\times SO(5))$ in
the $\mathbf{16}$ is $5$, namely the (real) dimension of the irrep. $\mathbf{10}$ of the $D=7$ $U$-duality group $SL(5,\mathbb{R})$, occurring in
(\ref{emb-1})-(\ref{emb-2}) (in this case, $d_{g}=12$).

In $D=7$, the relevant embedding in $\mathcal{N}=4$ supergravity reads:
\begin{eqnarray}
SL(5,\mathbb{R}) &\supset &_{s}^{\max }SL(3,\mathbb{R})\times SL(2,\mathbb{R})\times SO(1,1);  \label{emb-3} \\
\mathbf{10} &=&\left( \mathbf{3,2}\right) _{1}+\left( \mathbf{3}^{\prime }\mathbf{,1}\right) _{-4}+\left( \mathbf{1,1}\right) _{6},  \label{emb-4}
\end{eqnarray}
and the Iwasawa degree $d$ of the $D=7$ scalar manifold $G_{7,\mathbb{O}_{s}}/mcs\left( G_{7,\mathbb{O}_{s}}\right) =SL(5,\mathbb{R})/SO(5)$ in the
$\mathbf{10}$ is $6$, namely the (real) dimension of the irrep. $\left(\mathbf{3,2}\right) $ of the $D=8$ $U$-duality group $SL(3,\mathbb{R})\times
SL(2,\mathbb{R})$ occurring in (\ref{emb-3})-(\ref{emb-4}) (in this case, $d_{g}=6$).

In $D=8$, the relevant embedding in $\mathcal{N}=2$ supergravity is:
\begin{eqnarray}
SL(3,\mathbb{R})\times SL(2,\mathbb{R}) &\supset &_{s}^{\max }GL(2,\mathbb{R})\times SO(1,1);  \label{emb-5} \\
\left( \mathbf{3,2}\right) &=&2\cdot \left( \mathbf{2}_{1}+\mathbf{1}_{-2}\right) ,  \label{emb-6}
\end{eqnarray}
and the Iwasawa degree $d$ of the $D=8$ scalar manifold $G_{8,\mathbb{O}_{s}}/mcs\left( G_{8,\mathbb{O}_{s}}\right) =SL(3,\mathbb{R})/SO(3)\times
SL(2,\mathbb{R})/SO(2)$ in the $\left( \mathbf{3,2}\right) $ is $2$, namely the (real) dimension of the irrep. $\mathbf{2}$ of the $D=9$ $U$-duality
group $GL(2,\mathbb{R})$ occurring in (\ref{emb-5})-(\ref{emb-6}).

Finally, in $D=9$ the relevant embedding in $\mathcal{N}=2$ supergravity is:
\begin{eqnarray}
GL(2,\mathbb{R}) &\supset &_{s}^{\max }SO(1,1)_{IIA}\times SO(1,1); \label{emb-7} \\
\mathbf{2} &=&\mathbf{1}_{\alpha }+\mathbf{1}_{-\alpha },  \label{emb-8}
\end{eqnarray}
where the $SO(1,1)$-weight $\alpha $ in (\ref{emb-8}) depends on normalization, and the subscript \textquotedblleft $IIA$" in the first
$SO(1,1)$ in the r.h.s. of (\ref{emb-7}) characterizes it as the $U$-duality group of type IIA $D=10$ supergravity, and distinguishes it from the second
(Kaluza-Klein) $SO(1,1)$. The Iwasawa degree $d$ of the $D=9$ scalar manifold $G_{9,\mathbb{O}_{s}}/mcs\left( G_{9,\mathbb{O}_{s}}\right) =GL(2,\mathbb{R})/SO(2)$
in the $\mathbf{2}$ is $1$, namely the (real) dimension of the $SO(1,1)_{IIA}$-singlet $\mathbf{1}$ occurring in (\ref{emb-7})-(\ref{emb-8}).

As mentioned above, the representation determining the Iwasawa degree $d$ of the scalar manifold $G_{D,\mathbb{O}_{s}}/mcs(G_{D,\mathbb{O}_{s}})$ in $D$
dimensions (in the relevant $G_{D,\mathbb{O}_{s}}$-irrep.) is the $G_{D+1,\mathbb{O}_{s}}$-irrep. in which the $1$-form (Abelian) potentials of
maximal supergravity sit (\textit{cfr.} \textit{e.g.} \cite{Kleinschmidt,Brane Orbits}). In this sense, one can
understand why type IIB supergravity in $D=10$ (along with its $SL(2,\mathbb{R})$ $U$-duality) has not been considered : it has no $1$-form potentials.
As discussed in \cite{4-Kleinsch, 5-Kleinsch, Kleinschmidt}, the relevant $1$-form $U$-irrep. can be predicted by exploiting the infinite-dimensional
(Kac-Moody extended) algebra $E_{11}$. It would be interesting, \textit{at least} for maximal supergravity in $D=3,...,9$, to understand the matching
between the relevant Iwasawa degree $d$ of $\left( G_{D,\mathbb{O}_{s}}/mcs(G_{D,\mathbb{O}_{s}})\right) _{\mathbf{R}_{D}}$ and the (real)
dimension of the $G_{D+1,\mathbb{O}_{s}}$-irrep. $\mathbf{R}_{D+1}$ in terms of $E_{11}$; we leave this task for future further investigations.

\subsection{\label{sec:Universal}Universal Nilpotency Degree of Axionic Iwasawa Generators}

Considering symmetric cosets $\frac{G_{4}}{H_{4}}$ (\ref{M}) which are scalar manifolds of Maxwell-Einstein theories of (super)gravity in $D=4$
(Lorentzian) dimensions, it is here worth commenting on the relation between the (\textquotedblleft standard"; \textit{cfr.} Footnote 4)
total Iwasawa construction (\ref{SIC-M}):
\begin{equation}
\mathfrak{g}_{4}=\mathfrak{h}_{4}\oplus \mathfrak{a}_{4}\oplus \mathfrak{n}_{4},  \label{total-SIC}
\end{equation}%
exploited in the present paper, and the Iwasawa construction restricted to the \textit{axionic} generators of the $D=4$ $U$-duality group $G_{4}$,
recently studied in \cite{CFGM-d}. Note that in (\ref{total-SIC}) the sums are not direct, but they do respect the Lie algebra $\mathfrak{g}_{4}$;
$\mathfrak{h}_{4}$ denotes the Lie algebra of $H_{4}=mcs\left( G_{4}\right) $, while $\mathfrak{a}_{4}$ and $\mathfrak{n}_{4}$ respectively stand for the
(maximal) Abelian (Cartan) non compact sub-algebra (whose dimension equals the rank $\mathbf{r}$ of $\frac{G_{4}}{H_{4}}$; \textit{cfr.} Table~\ref{tab:fourierFg3}) and the maximal nilpotent set of nilpotent (namely, Iwasawa) generators, whose cardinality has been denoted by $\mathbf{I}$ in
Table~\ref{tab:fourierFg3}.

As illustrative examples, we will consider two theories : the magic exceptional $\mathcal{N}=2$ theory \cite{GST}, based on $J_{3}^{\mathbb{O}}$,
and the maximal $\mathcal{N}=8$ supergravity, based on $J_{3}^{\mathbb{O}_{s}}$.

\subsubsection{\label{Magic-Exceptional}$\mathcal{N}=2$, $D=4$ Magic Exceptional ($J_{3}^{\mathbb{O}}$)}

As yielded by (\ref{total-SIC}), the relevant (maximal, symmetric) embedding pertaining to the Iwasawa decomposition of the $D=4$ $U$-duality group
$G_{4}\equiv G_{4,\mathbb{O}}=E_{7(-25)}$ reads
\begin{eqnarray}
E_{7(-25)} &\supset &E_{6(-78)}\times U(1);  \label{decomp-1} \\
\underset{\mathfrak{g}_{4}}{\underbrace{\mathbf{133}}} &=&\underset{\mathfrak{h}_{4}}{\underbrace{\mathbf{78}_{0}+\mathbf{1}_{0}}}+
\underset{\mathfrak{a}_{4}\oplus \mathfrak{n}_{4}}{\underbrace{\mathbf{27}_{-2}+\overline{\mathbf{27}}_{2}}},  \label{decomp-2}
\end{eqnarray}
where $\dim _{\mathbb{R}}\left( \mathfrak{a}_{4}\right) =$rank$\left( \frac{E_{7(-25)}}{E_{6(-78)}\times U(1)}\right) =:\mathbf{r}=3$, and
$\dim _{\mathbb{R}}\left( \mathfrak{n}_{4}\right) =:\mathbf{I}=51$. The maximal degree of the polynomial in the Iwasawa decomposition of $\frac{E_{7(-25)}}{E_{6(-78)}\times U(1)}$ in the $\mathbf{56}$ of $E_{7(-25)}$ is $d=9$.

On the other hand, the Iwasawa parametrization of the axionic generators of $E_{7(-25)}$, performed in \cite{CFGM-d}, considers the maximal triangular
subgroup of $E_{7(-25)}$, containing the $D=5$ $U$-duality group $G_{5}\equiv G_{5,\mathbb{O}}=E_{6(-26)}$:
\begin{eqnarray}
E_{7(-25)} &\supset &E_{6(-26)}\times SO(1,1)_{KK}\times _{s}T_{27};
\label{decomp-3} \\
\underset{\mathfrak{g}_{4}}{\underbrace{\mathbf{133}}} &=&\underset{\mathfrak{g}_{5}}{\underbrace{\mathbf{78}_{0}}}
+\underset{\mathfrak{so}(1,1)_{KK}}{\underbrace{\mathbf{1}_{0}}}+\underset{T_{27}}{\underbrace{\mathbf{27}_{-2}}}+\mathbf{27}_{2}^{\prime },  \label{decomp-4}
\end{eqnarray}
where the \textquotedblleft $KK$" subscript denotes the Kaluza-Klein nature of the commuting $1$-dimensional factor, the subscript \textquotedblleft $s$"
indicates the semi-direct nature of the product, and $T_{27}\equiv \mathbb{R}^{27}$ corresponds to the $27$ axionic Iwasawa generators, in one-to-one
correspondence with the fifth component $A_{5}^{I}$ of the $27$ Abelian vector potentials of the corresponding theory in $D=5$; note that
(\ref{decomp-3}) is nothing but a different, non compact form of (\ref{decomp-1}). In this case, the maximal degree $d$ pertaining to the set of
axionic generators $T_{27}$ of the rank-$3$ special K\"{a}hler symmetric coset $\frac{E_{7(-25)}}{E_{6(-78)}\times U(1)}$ is $d_{axionic}=3$~\cite{dWVVP, CFGM-d, univ-nilp-Refs, ADFLL}. This is \textit{universal}, namely it concerns the axionic generators of the global electric-magnetic isometries
of \textit{all} $D=4$ theories which enjoy a $D=5$ uplift, even in the case of non-symmetric nor homogeneous (vector multiplets') scalar manifolds.

At the level of the coset manifolds, (\ref{decomp-3}) and (\ref{decomp-4}) can be reformulated as the isometry
\begin{eqnarray}
\frac{E_{7(-25)}}{E_{6(-78)}\times U(1)} &\sim &\frac{E_{6(-26)}}{F_{4(-52)}}\times SO(1,1)_{KK}\times _{s}T_{27};  \label{decomp-5} \\
54 &=&26+1+27,  \label{decomp-6}
\end{eqnarray}
where the numbers in the second line denote the real dimensions of the manifolds. For theories based on \textit{Freudenthal triple systems} $\mathbf{F}\left( J_{3}\right) $ on
rank-$3$ (simple or semi-simple) Jordan algebras $J_{3}$, (\ref{decomp-3}) and (\ref{decomp-4}) can be $q$-parametrized as follows:
\begin{eqnarray}
G_{\mathbf{F}\left( J_{3}\right) ,4} &\supset &G_{J_{3},5}\times
SO(1,1)_{KK}\times _{s}T_{3q+3};  \label{decomp-7} \\
\underset{\mathfrak{g}_{\mathbf{F}\left( J_{3}\right) ,4}}{\underbrace{\mathbf{Adj}\left( G_{F\left( J_{3}\right) ,4}\right) }}
&=&\underset{\mathfrak{g}_{J_{3},5}}{\underbrace{\mathbf{Adj}\left( G_{J_{3},5}\right) }}+\underset{\mathfrak{so}(1,1)_{KK}}{\underbrace{\mathbf{1}_{0}}}
+\underset{T_{3q+3}}{\underbrace{\left( \mathbf{3q+3}\right) _{-2}}}+\left( \mathbf{3q+3}\right) _{2}^{\prime },  \label{decomp-8}
\end{eqnarray}
where $G_{\mathbf{F}\left( J_{3}\right) ,4}=Aut\left( \mathbf{F}\left(J_{3}\right) \right) =Conf\left( J_{3}\right) $ and $G_{J_{3},5}=Str_{0}
\left( J_{3}\right) $ respectively denote the $D=4$ and $D=5$ $U$-duality groups; as mentioned above, for simple rank-$3$ Jordan algebras,
$q:=\dim _{\mathbb{R}}\mathbb{A}=\dim _{\mathbb{R}}\mathbb{A}_{s}$, whereas $q=\left(m+n-4\right) /3$ for the class of semi-simple Jordan algebras
$\mathbb{R\oplus }\mathbf{\Gamma }_{m-1,n-1}$, where $\mathbf{\Gamma }_{m-1,n-1}$ stands for the Clifford algebra of $O(m-1,n-1)$. It
should be remarked that, since this is a $\mathcal{N}=2$, $D=4$ theory, $G_{J_{3},5}$ is nothing but a non compact, real form of the $mcs\left(
G_{\mathbf{F}\left( J_{3}\right) ,4}\right) $.

Note that the \textquotedblleft standard" total Iwasawa construction (\ref{SIC-M}) (or, equivalently, (\ref{total-SIC})) takes into account of the
maximal number $6q+3$ of Iwasawa (nilpotent) generators, but with high maximal degree ($d=9$ for all $J_{3}^{\mathbb{A}}$-based magic
theories in $D=10$; \textit{cfr.} Sec. \ref{q-indep}), while the Iwasawa construction exploited in \cite{CFGM-d} is restricted to the axionic
generators of the $D=4$ $U$-duality group $G_{4}$, but with lower (and \textit{universal}) maximal degree ($d_{axionic}=3$). \textit{At least}
in the theories related to rank-$3$ Jordan algebras, the $3q+3$ axionic generators are related to the whole set of $6q+3$ Iwasawa generators through
some Wick's rotation and linear combination: in fact, in the magic exceptional supergravity $E_{6(-78)}\times U(1)$ gets converted into
$E_{6(-26)}\times SO(1,1)$.

\subsubsection{\label{Maximal}$\mathcal{N}=8$, $D=4$ Maximal ($J_{3}^{\mathbb{O}_{s}}$)}

The difference between the two aforementioned approaches to the Iwasawa construction is more striking in the maximal supergravity theory.

As yielded by (\ref{total-SIC}), the relevant (maximal, symmetric) embedding pertaining to the Iwasawa decomposition of the $D=4$ $U$-duality group
$G_{4}\equiv G_{4,\mathbb{O}_{s}}=E_{7(7)}$ reads
\begin{eqnarray}
E_{7(7)} &\supset &SU(8);  \label{decomp-1-max} \\
\underset{\mathfrak{g}_{4}}{\underbrace{\mathbf{133}}} &=&\underset{\mathfrak{h}_{4}}{\underbrace{\mathbf{63}_{0}}}+
\underset{\mathfrak{a}_{4}\oplus \mathfrak{n}_{4}}{\underbrace{\mathbf{70}}},  \label{decomp-2-max}
\end{eqnarray}
where $\dim _{\mathbb{R}}\left( \mathfrak{a}_{4}\right) =$rank$\left( \frac{E_{7(7)}}{SU(8)}\right) =:\mathbf{r}=7$, and $\dim _{\mathbb{R}}\left(
\mathfrak{n}_{4}\right) =:\mathbf{I}=63$. Note that in this case the coset rank $\mathbf{r}$ matches the group rank $l$ as well as the coset character $\chi $;
this is a common feature of all theories based on \textit{split} algebras $\mathbb{A}_{s}$ (resulting in \textit{split}, \textit{i.e.}
maximally non compact, forms of the corresponding $U$-duality groups, as it is the case for $E_{7(7)}$). The maximal degree of the polynomial in the Iwasawa decomposition of $\frac{E_{7(7)}}{SU(8)}$ in the $\mathbf{56}$ of $E_{7(7)}$ is $d=27$.

On the other hand, the Iwasawa parametrization of axionic generators of $E_{7(7)}$, performed in \cite{CFGM-d}, considers the maximal triangular
subgroup of $E_{7(7)}$, containing the $D=5$ $U$-duality group $G_{5}\equiv G_{5,\mathbb{O}_{s}}=E_{6(6)}$ :
\begin{eqnarray}
E_{7(7)} &\supset &E_{6(6)}\times SO(1,1)_{KK}\times _{s}T_{27};\label{decomp-3-max} \\
\underset{\mathfrak{g}_{4}}{\underbrace{\mathbf{133}}} &=&\underset{\mathfrak{g}_{5}}{\underbrace{\mathbf{78}_{0}}}+
\underset{\mathfrak{so}(1,1)_{KK}}{\underbrace{\mathbf{1}_{0}}}+\underset{T_{27}}{\underbrace{\mathbf{27}_{-2}}}+\mathbf{27}_{2}^{\prime }.  \label{decomp-4-max}
\end{eqnarray}
Once again, note that (\ref{decomp-3-max}) is nothing but a non compact form of the maximal and symmetric embedding (\ref{decomp-1}) (different from the
non compact form (\ref{decomp-3})). In this case, the maximal degree pertaining to the set of axionic generators $T_{27}$ of
$\frac{E_{7(7)}}{SU(8)}$ is $d_{axionic}=3$; it is the same as the one of $\frac{E_{7(-25)}}{E_{6(-78)}\times U(1)}$, because, as mentioned above, it is
\textit{universal}. In particular, at the level of cosets, (\ref{decomp-3-max}) and (\ref{decomp-4-max}) specify as
\begin{eqnarray}
\frac{E_{7(7)}}{SU(8)} &\sim &\frac{E_{6(6)}}{USp(8)}\times SO(1,1)_{KK}\times _{s}T_{27};  \label{decomp-5-max} \\
70 &=&42+1+27.  \label{decomp-6-max}
\end{eqnarray}
It should be remarked that, since this is \textit{not} an $\mathcal{N}=2$, $D=4$ theory, $G_{J_{3},5}$ is \textit{not} a non compact, real form of the
$mcs\left( G_{\mathbf{F}\left( J_{3}\right) ,4}\right) $.

In this case, the \textquotedblleft standard" total Iwasawa construction (\ref{SIC-M}) (or, equivalently, (\ref{total-SIC})) takes into account of the
whole set of $63$ Iwasawa (nilpotent) generators, but with high maximal degree $d=27$ (for the other $\mathbb{H}_{s}$- and
$\mathbb{C}_{s}$-based theories, see Sec. \ref{q-dep}), while the Iwasawa construction exploited in \cite{CFGM-d} is restricted to the $27$ axionic generators of
the $D=4$ $U$-duality group $E_{7(7)}$, but with lower (and \textit{universal}) maximal degree $d_{axionic}=3$. As observed above, the $27$ axionic
generators are related to the whole set of $63$ Iwasawa generators through some Wick's rotation and linear combination.

\subsubsection{\label{T^3-Model}The $\mathcal{N}=2$, $D=4$ $T^{3}$ Model}

As already mentioned, among the theories with symmetric (vector multiplets') scalar manifolds, there is a unique case in which the two treatments under
consideration coincide : it is the so-called $T^{3}$ model of $\mathcal{N}=2$, $D=4$ Maxwell-Einstein supergravity, in which the unique Abelian vector
multiplet is coupled \textit{non-minimally}, but rather through a cubic holomorphic prepotential $\mathcal{F}=T^{3}$, to the gravity multiplet.

In this case, (\ref{total-SIC}) trivially yields the (maximal, symmetric) embedding pertaining to the Iwasawa decomposition of the $D=4$ $U$-duality
group $G_{4}=SL(2,\mathbb{R})$:
\begin{eqnarray}
SL(2,\mathbb{R}) &\supset &U(1);  \label{decomp-1-T^3} \\
\underset{\mathfrak{g}_{4}}{\underbrace{\mathbf{3}}} &=&\underset{\mathfrak{h}_{4}}{\underbrace{\mathbf{1}_{0}}}+\underset{\mathfrak{a}_{4}\oplus
\mathfrak{n}_{4}}{\underbrace{\mathbf{1}_{2}+\mathbf{1}_{-2}}},
\label{decomp-2-T^3}
\end{eqnarray}
where $\dim _{\mathbb{R}}\left( \mathfrak{a}_{4}\right) =$rank$\left( \frac{SL(2,\mathbb{R})}{U(1)}\right) =:\mathbf{r}=1$, and
$\dim _{\mathbb{R}}\left( \mathfrak{n}_{4}\right) =:\mathbf{I}=1$. Note that also in this case the coset rank $\mathbf{r}$ matches the group rank $l$
as well as the coset character $\chi $, because $SL(2,\mathbb{R})$ trivially is the split form of $SU(2)$.

The maximal degree of the polynomial pertaining to the Iwasawa decomposition of $\frac{SL(2,\mathbb{R})}{U(1)}$ in the the relevant irrep. $\mathbf{4}$ (spin $s=3/2$) of $SL(2,\mathbb{R})$ is
$d=3 $. Thus, it matches the universal result $d_{axionic}=3$, which does pertain to the Iwasawa parametrization of the unique axionic generator of
$SL(2,\mathbb{R})$ \cite{CFGM-d}, considering the maximal triangular subgroup of $SL(2,\mathbb{R})$ itself:
\begin{eqnarray}
SL(2,\mathbb{R}) &\supset &SO(1,1)_{KK}\times _{s}T_{1};\label{decomp-3-T^3} \\
\underset{\mathfrak{g}_{4}}{\underbrace{\mathbf{3}}} &=&\underset{\mathfrak{so}(1,1)_{KK}}{\underbrace{\mathbf{1}_{0}}}+
\underset{T_{1}}{\underbrace{\mathbf{1}_{-2}}}+\mathbf{1}_{2}.  \label{decomp-4-T^3}
\end{eqnarray}
Once again, note that (\ref{decomp-3-T^3}) is nothing but a non compact form of the maximal and symmetric embedding (\ref{decomp-1}). In this case the
$D=5$ $U$-duality group is empty ($G_{5}=\varnothing $), because the $T^{3}$-model is the unique model which uplifts to \textit{\textquotedblleft pure"}
minimal ($\mathcal{N}=2$) $D=5$ theory, in which only the gravity multiplet is present.

In particular, at the level of cosets, (\ref{decomp-3-T^3}) and (\ref{decomp-4-T^3}) respectively specify as
\begin{eqnarray}
\frac{SL(2,\mathbb{R})}{U(1)} &\sim &SO(1,1)_{KK}\times _{s}T_{1};\label{decomp-5-T^3} \\
2 &=&1+1.  \label{decomp-6-T^3}
\end{eqnarray}

The matching of the two approaches can be explained simply by observing that the unique generator of $\mathfrak{n}_{4}$ in (\ref{decomp-2-T^3}) is $T_{1}$
occurring in (\ref{decomp-5-T^3}).

\enlargethispage{5\baselineskip}

\begin{table}[H]
\begin{center}
\begin{tabular}{|c|ccc|ccc|ccc|}
\hline
$%
\begin{array}{c}
\\
Special~Kaehler \\
~Symmetric~Space \\
~
\end{array}
$ &  & $\mathbf{R}$ &  &  & $d$ &  &  & $d_{g}$ &  \\ \hline\hline $
\begin{array}{c}
\\
\overline{\mathbb{CP}}^{n}:\frac{SU(1,n)}{SU(n)\times U\left( 1\right) },~~n\in \mathbb{N} \\
~
\end{array}
$ &  & $\mathbf{1+n}=(1,0,...,0)$ &  &  & $\begin{array}{ll} 2&  (n>1) \\ 1& (n=1) \end{array}$ &  &  & $\begin{array}{ll} 2& (n>1) \\ 0& (n=1) \end{array}$ &  \\ \hline $
\begin{array}{c}
\\
\frac{SU(1,1)}{U\left( 1\right) }\times \frac{SO(2,n)}{SO(n)\times U\left( 1\right) }, \\
~ \\
n\in \mathbb{N}~~\left( \mathbb{R}\oplus \mathbf{\Gamma }_{n-1,1}\right) \\
~
\end{array}
$ &  & $\begin{array}{ll} \mathbf{(2,2+n)}=\\(1)(1,0,...,0) &(n>1)\\
(1)(2) &(n=1)\end{array}$ &  &  & $\begin{array}{c} 5\  (n>2) \\ 3\ (n=1,2) \end{array}$ &  &  &
$\begin{array}{ll} 4& (n>2) \\ 0 & (n=1,2)  \end{array}$ &  \\
\hline
$
\begin{array}{c}
\\
\frac{SU(1,1)}{U\left( 1\right) }~~\left( \mathbb{R}\right) \\
~
\end{array}
$ &  & $\mathbf{4}=(3)$ &  &  & $3$ &  &  & $0$ &  \\ \hline
$
\begin{array}{c}
\\
\frac{Sp(6,\mathbb{R})}{SU\left( 3\right) \times U\left( 1\right) }~~\left( J_{3}^{\mathbb{R}}\right) \\
~
\end{array}
$ &  & $\mathbf{14^{\prime }}=(001)$ &  &  & $9$ &  &  & $8$ &  \\ \hline
$
\begin{array}{c}
\\
\frac{SU(3,3)}{SU\left( 3\right) \times SU\left( 3\right) \times U\left(1\right) }~~\left( J_{3}^{\mathbb{C}}\right) \\
~
\end{array}
$ &  & $\mathbf{20}=(00100)$ &  &  & $9$ &  &  & $8$ &  \\ \hline
$
\begin{array}{c}
\\
\frac{SO^{\ast }(12)}{SU\left( 6\right) \times U\left( 1\right) }~~\left(J_{3}^{\mathbb{H}},\mathcal{N}=2\Leftrightarrow \mathcal{N}=6\right) \\
~
\end{array}
$ &  & $
\begin{array}{c}
\\
\mathbf{32}=(000010) \\
~or~\mathbf{32^{\prime }}=(000001)
\end{array}
$ &  &  & $9$ &  &  & $8$ &  \\ \hline
$
\begin{array}{c}
\\
\frac{E_{7\left( -25\right) }}{E_{6}\times SO\left( 2\right) }~~\left(J_{3}^{\mathbb{O}}\right) \\
~
\end{array}
$ &  & $\mathbf{56}=(0000010)$ &  &  & $9$ &  &  & $8$ &  \\ \hline
\end{tabular}
\end{center}
\caption{$\mathcal{N}=2$, $D=4$ symmetric special K\"{a}hler vector multiplets' scalar manifolds. If any, the related Euclidean rank-$3$ Jordan
algebras are reported in brackets throughout (the notation of \protect\cite{LA08} is used). By defining $q:=dim_{\mathbb{R}}\mathbb{A}$ ($=1,2,4,8$ for
$\mathbb{A}=\mathbb{R},\mathbb{C},\mathbb{H},\mathbb{O}$ respectively), the \textit{complex} dimension of the $\mathcal{N}=2,d=4$ symmetric special
K\"{a}hler manifolds based on $J_{3}^{\mathbb{A}}$ is $3q+3$ \protect\cite{Ferrara-Marrani-2}. The irrep. $\mathbf{R}$ of $G_{nc}$ relevant for
supergravity is reported, along with the corresponding degrees $d$ and $d_{g} $. Note that $d=9$ and $d_{g}=8$ for every $q=1,2,4,8$. For the Dynkin notation of
irreps., we adopt the conventions of \protect\cite{Slansky} throughout.
Notice that in the sixth row for both the representations $\mathbf{32}$ and $\mathbf{32'}$ of $SO^*(12)$ the degree is 9. This is due to the fact that the highest weight corresponds to the non compact root (see Remark after Table~\ref{tab:degrees}).}\label{tab:SKG-D=4}
\end{table}


\begin{table}[H]
\begin{center}
\begin{tabular}{|c|ccc|ccc|ccc|}
\hline
$
\begin{array}{c}
\\
Quaternionic \\
Symmetric~Space \\
~
\end{array}
$ &  & $\mathbf{R}=\mathbf{Adj}(G)$ &  &  & $d$ &  &  & $d_{g}$ &  \\
\hline\hline
$
\begin{array}{c}
\\
\frac{SU(2,n+1)}{SU(n+1)\times SU\left( 2\right) \times U\left( 1\right) },~~n\in \mathbb{N}~\cup \left\{ 0\right\} \\
~
\end{array}
$ &  & $\begin{array}{c} \left( \mathbf{n+3}\right) ^{2}\mathbf{-1} \\=\left( 1,0,....0,1\right) \end{array}
$ &  &  & $\begin{array}{c} 8\ (n>1) \\ 6\ (n=1) \\ 4\ (n=0) \end{array}$ &  &  &
$\begin{array}{c} 6\ (n>1) \\ 4\ (n=1) \\ 2\ (n=0) \end{array}$ &  \\ \hline
$
\begin{array}{c}
\\
\frac{SO(4,n+2)}{SO(n+2)\times SO\left( 4\right) }, \\ ~ \\
~n\in \mathbb{N}\cup \left\{ 0,-1\right\} ~~\left( \mathbb{R}\oplus \mathbf{\Gamma }_{n-1,1}\right) \\
~
\end{array}
$ &  & $
\begin{array}{l}
\frac{\left( \mathbf{n+6}\right) \left( \mathbf{n+5}\right) }{\mathbf{2}} \\ =(0,1,0,...0)
\end{array}
$ &  &  & $\begin{array}{l} 14\ (n>2) \\ 10\ (n=2,1) \\ 6\ (n=0) \\ 2\ (n=-1) \end{array}$ &  &  &
$\begin{array}{l} 12\ (n>2) \\ 8\ (n=2,1) \\ 4\ (n=0) \\ 0\ (n=-1) \end{array}$ &  \\ \hline
$
\begin{array}{c}
\\
\frac{G_{2\left( 2\right) }}{SO\left( 4\right) }~~\left( \mathbb{R}\right)
\\
~
\end{array}
$ &  & $\mathbf{14}=(10)$ &  &  & $10$ &  &  & $8$ &  \\ \hline $
\begin{array}{c}
\\
\frac{F_{4\left( 4\right) }}{USp\left( 6\right) \times SU\left( 2\right) } ~~\left( J_{3}^{\mathbb{R}}\right) \\
~
\end{array}
$ &  & $\mathbf{52}=\left( 1000\right) $ &  &  & $22$ &  &  & $20$ &  \\
\hline
$
\begin{array}{c}
\\
\frac{E_{6\left( 2\right) }}{SU(6)\times SU\left( 2\right) }~~\left( J_{3}^{\mathbb{C}}\right) \\
~
\end{array}
$ &  & $\mathbf{78}=\left( 000001\right) $ &  &  & $22$ &  &  & $20$ &  \\
\hline
$
\begin{array}{c}
\\
\frac{E_{7\left( -5\right) }}{SO(12)\times SU\left( 2\right) }~~\left(J_{3}^{\mathbb{H}},\mathcal{N}=4\Leftrightarrow \mathcal{N}=12\right) \\
~
\end{array}
$ &  & $\mathbf{133}=\left( 1000000\right) $ &  &  & $22$ &  &  & $20$ &  \\
\hline
$
\begin{array}{c}
\\
\frac{E_{8\left( -24\right) }}{E_{7}\times SU\left( 2\right) }~~\left(J_{3}^{\mathbb{O}}\right) \\
~
\end{array}
$ &  & $\mathbf{248}=\left( 00000010\right) $ &  &  & $22$ &  &  & $20$ & \\ \hline
\end{tabular}
\end{center}
\caption{$\mathcal{N}=4$, $D=3$ symmetric quaternionic spaces, obtained through $c$-map \protect\cite{CFG} from the special K\"{a}hler manifolds of
Table~\protect\ref{tab:SKG-D=4}. Since the hypermultiplets are insensitive to dimensional reduction, these spaces can also be regarded as the
hypermultiplets' scalar manifolds in $D=3,4,5,6$ (for a discussion of anomaly-free conditions in minimal chiral $(1,0)$ theories in $D=6$, see
\textit{e.g.} \protect\cite{AFMT-1}). The relevant $G$ -irrep. is the adjoint. In general, starting from a special K\"{a}hler
geometry with $\dim _{\mathbb{C}}=n$, the $c$-map generates a quaternionic manifold with $\dim _{\mathbb{H}}$ $=n+1$ \protect\cite{CFG}. Thus, the
quaternionic dimensions of the spaces corresponding to $q=1,2,4,8$ is $3q+4$ \protect\cite{CFG,Ferrara-Bianchi}. The irrep.
$\mathbf{R}=\mathbf{Adj}(G_{nc})$ relevant for supergravity is reported, along with the corresponding degrees $d$ and $d_{g}$. Note that $d=22$ and $d_{g}=20$ for
every $q=1,2,4,8$.
The unique other class of symmetric quaternionic spaces is given by the non compact quaternionic projective spaces
$\overline{\mathbb{HP}}^{n}=Usp\left( 2,2n)/(Usp(2)\times Usp(2n)\right) $ ($n\in \mathbb{N}$),
which is not in the $c$-map image of any (symmetric) special K\"{a}hler manifold \protect\cite{CFG} (indeed, its $\Omega $-tensor - and thus, its
corresponding quaternionic potential - identically vanishes \protect\cite{Ogievetsky-et-al}) for this class, still
$\mathbf{R}(Usp(2,2n))$=$\mathbf{Adj}=(\mathbf{2n+3})(\mathbf{n+1})$=(2,0,...,0), and $d=4$, $d_{g}=2$ if $n>1$, $d=2$, $d_{g}=0$ if $n=1$.}
\label{tab:HKG-D=3}
\end{table}


\begin{table}[H]
\begin{center}
\begin{tabular}{|c|ccc|ccc|ccc|}
\hline
$
\begin{array}{c}
\\
Real~Special \\
Symmetric~Space \\
~
\end{array}
$ &  & $\mathbf{R}$ &  &  & $d$ &  &  & $d_{g}$ &  \\ \hline\hline
$
\begin{array}{c}
\\
\frac{SO(1,n)}{SO(n)},~~n\geq 2 \\
~
\end{array}
$ &  & $\begin{array}{rll} 
\mathbf{1+n}&=(1,0,...,0) & (n>2) \\ \mathbf{3}&=(2) & (n=2)\end{array}$ &  &  & $2$ &  &  & $0$ &  \\ \hline
$
\begin{array}{c}
\\
SO(1,1)\otimes \frac{SO(1,n-1)}{SO(n-1) }, \\
~ \\
n\ge 2~~\left( \mathbb{R}\oplus \mathbf{\Gamma }_{n-1,1}\right) \\
~
\end{array}
$ &  & $
\begin{array}{l}
\left( \mathbf{1},\mathbf{1}\right) +\left( \mathbf{1},\mathbf{n}\right) =
\\
\left( 0\right) \left( 0\right) +\left( 0\right) \left( 1,0,...,0\right) \, (n>3)\\
\left( 0\right) \left( 0\right) +\left( 0\right) \left(2\right) \, (n=3)\\
\left( 0\right) \left( 0\right) +\left( 0\right) \left(0\right) \, (n=2)
\end{array}
$ &  &  & $\begin{array}{c} 2\ (n>2) \\ 0\ (n=2) \end{array}$  &  &  & $0$ &  \\ \hline
$
\begin{array}{c}
\\
\frac{SL(3,\mathbb{R})}{SO(3)}~~\left( J_{3}^{\mathbb{R}}\right) \\
~
\end{array}
$ &  & $\mathbf{6^{\prime }}=(02)$ &  &  & $4$ &  &  & $2$ &  \\ \hline
$
\begin{array}{c}
\\
\frac{SL(3,\mathbb{C})}{SU(3)}~~\left( J_{3}^{\mathbb{C}}\right) \\
~
\end{array}
$ &  & $\mathbf{9}=\left( \mathbf{3},\overline{\mathbf{3}}\right) =(10)(01)$ &  &  & $4$ &  &  & $2$ &  \\ \hline
$
\begin{array}{c}
\\
\frac{SU^{\ast }(6)}{USp(6)}~~\left( J_{3}^{\mathbb{H}},\mathcal{N}=2\Leftrightarrow \mathcal{N}=6\right) \\
~
\end{array}
$ &  & $
\begin{array}{c}
\\
\mathbf{15}=(01000)
\end{array}
$ &  &  & $4$ &  &  & $2$ &  \\ \hline
$
\begin{array}{c}
\\
\frac{E_{6(-26)}}{F_{4}}~~\left( J_{3}^{\mathbb{O}}\right) \\
~
\end{array}
$ &  & $\mathbf{27}=(100000)$ &  &  & $4$ &  &  & $2$ &  \\ \hline
\end{tabular}
\end{center}
\caption{$\mathcal{N}=2$, $D=5$ symmetric real special vector multiplets' scalar manifolds; they are nothing but the \textit{moduli spaces }of non-BPS
$Z\neq 0$ extremal black hole attractors of the corresponding theory in $D=4$ \protect\cite{Ferrara-Marrani-2}. They are obtained from the corresponding symmetric special K\"{a}hler manifolds of Table~\protect\ref{tab:SKG-D=4} through the so-called $R$-map, with the notable exception of the case $\frac{SO(1,n)}{SO(n)}$. From \cite{dWvP} the latter does not correspond to the $R$-map image of any symmetric space, as the \textit{minimally coupled} $D=4$ models based on $\overline{\mathbb{CP}}^{n}$ cannot be uplifted to $D=5$. If any, the related Euclidean rank-$3$ Jordan
algebras are reported in brackets throughout. The \textit{real} dimension of the manifolds based on $J_{3}^{\mathbb{A}}$ is $3q+2$
\protect\cite{Ferrara-Marrani-2}. The irrep. $\mathbf{R}$ of $G_{nc}$ relevant for supergravity is reported, along with the corresponding degrees $d$ and
$d_{g} $. Note that $d=4$ and $d_{g}=2$ for every $q=1,2,4,8$. The coset $\frac{SO(1,n)}{SO(n)}$ is not related to a rank-$3$ Jordan algebra, and the corresponding $D=4$ theory is based on an homogeneous non-symmetric manifold; for further details, see \textit{e.g.} \protect\cite{VPdW-1} (see also \protect\cite{Small-Orbits}). $\frac{SL(3,\mathbb{C})}{SU(3)}$ is the unique coset (\textit{at least} among symmetric scalar manifolds in supergravity) belonging to the class of symmetric spaces $\frac{G_{\mathbb{C}}}{G_{\mathbb{R}}}$
\protect\cite{Helgason}, where\ $G_{\mathbb{C}}$ is a complex (non compact ) (semi-)simple Lie group regarded as a real group, and $G_{\mathbb{R}}$ is
its compact, real form ($mcs\left( G_{\mathbb{C}}\right) =G_{\mathbb{R}}$); in general, $\frac{G_{\mathbb{C}}}{G_{\mathbb{R}}}$ is a Riemann symmetric
space with dim$_{\mathbb{R}}\left(\frac{G_{\mathbb{C}}}{G_{\mathbb{R}}}\right)$ $=$ dim$_{\mathbb{R}}\left( G_{\mathbb{R}}\right) $, and rank$\left(\frac{G_{\mathbb{C}}}{G_{\mathbb{R}}}\right)=$rank$(G_{\mathbb{R}})$; this case deserves a separate analysis, which will be given elsewhere (we here anticipate the result relevant for the present treatment). Note that the $\mathbf{6}^{\prime }$
of $SL(3,\mathbb{R})$ is not fundamental; together with the spin $s=3/2$ irrep. $\mathbf{4}$ of $SL(2,\mathbb{R})$ in the so-called $D=4$
$T^{3}$ model, this is the only case in which the relevant $G_{nc}$-repr. is not fundamental.}
\label{tab:RSG-D=5}
\end{table}

\newpage
\enlargethispage{7.2\baselineskip}
\vspace*{-4.5em}

\begin{table}[H]
\begin{center}
\begin{tabular}{|c|ccc|ccc|ccc|}
\hline
$
\begin{array}{c}
\\
D=4:G_{\mathcal{N}}/H_{\mathcal{N}} \\
~
\end{array}
$ &  & $\mathbf{R}$ &  &  & $d$ &  &  & $d_{g}$ &  \\ \hline\hline
$
\begin{array}{c}
\\
\mathcal{N}=3:\frac{SU(3,n)}{SU(3)\times SU\left( n\right) \times U\left(1\right) },~~n\in \mathbb{N} \\
~
\end{array}
$ &  & $\mathbf{3+n}=(1,0,...,0)$ &  &  & $\begin{array}{c} 6\ (n>3) \\ 5\ (n=3) \\ 4\ (n=2)\\ 2\ (n=1)  \end{array}$ &  &  &
$\begin{array}{c} 10\ (n>3) \\ 8\ (n=3) \\ 4\ (n=2)\\ 2\ (n=1)  \end{array}$ &  \\ \hline
$
\begin{array}{c}
\\
\mathcal{N}=4:\frac{SU(1,1)}{U\left( 1\right) }\times \frac{SO(6,n)}{SO(6)\times SO\left( n\right) }, \\
~ \\
n\in \mathbb{N} ~~\left( \mathbb{R}\oplus \mathbf{\Gamma }_{n-1,5}\right) \\
~
\end{array}
$ &  & $\begin{array}{l} \mathbf{(2,6+n)}\\=(1)(1,0,...,0) \end{array}$ &  &  &
$\begin{array}{l} 13\ (n>6) \\ 11\ (n=6,5) \\ 9\ (n=4)\\ 7\ (n=3)\\ 5\ (n=2)\\ 3\ (n=1)\end{array}$ &  &  &
$\begin{array}{l} 20\ (n>6) \\ 16\ (n=6,5) \\ 12\ (n=4)\\ 8\ (n=3)\\ 4\ (n=2)\\ 0\ (n=1) \\\end{array}$ &  \\
\hline
$
\begin{array}{c}
\\
\mathcal{N}=5:\frac{SU\left( 1,5\right) }{SU\left( 5\right) \times U\left(1\right) }~\left( M_{1,2}\left( \mathbb{O}\right) \right) \\
~
\end{array}
$ &  & $\mathbf{20}=(00100)$ &  &  & $2$ &  &  & $2$ &  \\ \hline
$
\begin{array}{c}
\\
\mathcal{N}=2\Leftrightarrow \mathcal{N}=6:\frac{SO^{\ast }\left( 12\right)}{SU\left( 6\right) \times U\left( 1\right) }~\left( J_{3}^{\mathbb{H}}\right) \\
~
\end{array}
$ &  & $
\begin{array}{c}
\\
\mathbf{32}=(000010) \\
~or~\mathbf{32^{\prime }}=(000001)
\end{array}
$ &  &  & $9$ &  &  & $8$ &  \\ \hline
$
\begin{array}{c}
\\
\mathcal{N}=8:\frac{E_{7\left( 7\right) }}{SU\left( 8\right) }~~\left(J_{3}^{\mathbb{O}_{s}}\right) \\
~
\end{array}
$ &  & $
\begin{array}{c}
\\
\mathbf{56}=(00000010)
\end{array}
$ &  &  & $27$ &  &  & $32$ &  \\ \hline
\end{tabular}
\end{center}
\caption{Scalar manifolds of $\mathcal{N}\geq 3$, $D=4$ supergravity theories. Note that the scalar manifolds of $\mathcal{N}=6$ theory coincides
with the one of $\mathcal{N}=2$ magic supergravity based on $J_{3}^{\mathbb{H}}$; in fact, these theories share the very same bosonic sector
\protect\cite{twin-Refs}. $M_{1,2}\left( \mathbb{O}\right) $ is an exceptional Jordan triple system, generated by $2\times 1$ vectors over $\mathbb{O}$
\protect\cite{GST}. Note that, despite the different relevant irrep. $\mathbf{R}$ ($\mathbf{6}$ \textit{vs.} rank-$3$ self-dual real $\mathbf{20}$),
the degree of the Iwasawa polynomial pertaining to $\mathcal{N}=2$ supergravity minimally coupled to $5$ vector multiplets (based on
$\overline{\mathbb{CP}}^{5}$) is the same as the one of $\mathcal{N}=5$ theory, which is \textit{pure} (no matter coupling allowed); indeed, both $\mathbf{6}$
and $\mathbf{20}$ are fundamental irreps. of $SU(1,5)$, and the string $\left( d_{1},...d_{l}\right) $ for $SU(1,n)$ is constant:
$\left(d_{1},...d_{l}\right) =2\left( 1,...,1\right) $.
Notice that in the fourth row for both the representations $\mathbf{32}$ and $\mathbf{32'}$ of $SO^*(12)$ the degree is 9. This is due to the fact that the highest weight corresponds to the non compact root (see Remark after Table~\ref{tab:degrees}).
} \label{tab:D=4}
\end{table}


\vspace{-4em}
\begin{table}[H]
\begin{center}
\begin{tabular}{|c|ccc|ccc|ccc|}
\hline
$
\begin{array}{c}
\\
D=5:G_{\mathcal{N}}/H_{\mathcal{N}} \\
~
\end{array}
$ &  & $\mathbf{R}$ &  &  & $d$ &  &  & $d_{g}$ &  \\ \hline\hline
$
\begin{array}{c}
\\
\mathcal{N}=4:
\\
SO(1,1)\times \frac{SO(5,n-1)}{SO(n-1)\times SO(5) } , \\
~ \\
n\ge 2~~\left( \mathbb{R}\oplus \mathbf{\Gamma }_{n-1,5}\right) \\
~
\end{array}
$ &  & $
\begin{array}{l}
\left( \mathbf{1},\mathbf{1}\right) +\left( \mathbf{1},\mathbf{n+4}\right) =
\\
\left( 0\right) \left( 0\right) +\left( 0\right) \left( 1,0,...,0\right)
\end{array}
$ &  &  & $\begin{array}{l} 10\ (n>5) \\ 8\ (n=5) \\ 6\ (n=4)\\ 4\ (n=3)\\ 2\ (n=2)   \end{array}$ &  &  &
$\begin{array}{l} 16\ (n>5) \\ 12\ (n=5) \\ 8\ (n=4)\\ 4\ (n=3)\\ 0\ (n=2) 
\end{array}$ &  \\ \hline
$
\begin{array}{c}
\\
\mathcal{N}=2\Leftrightarrow \mathcal{N}=6:\frac{SU^{\ast }(6)}{USp(6)}~\left( J_{3}^{\mathbb{H}}\right) \\
~
\end{array}
$ &  & $\mathbf{15}=(01000)$ &  &  & $4$ &  &  & $2$ &  \\ \hline
$
\begin{array}{c}
\\
\mathcal{N}=8:\frac{E_{6\left( 6\right) }}{USp\left( 8\right) }~\left(J_{3}^{\mathbb{O}_{s}}\right) \\
~
\end{array}
$ &  & $\mathbf{27}=(100000)$ &  &  & $16$ &  &  & $20$ &  \\ \hline
\end{tabular}
\end{center}
\caption{Scalar manifolds of $\mathcal{N}\geq 4$, $D=5$ supergravity theories. Note that the scalar manifold of the $\mathcal{N}=6$ theory coincides
with the one of $\mathcal{N}=2$ magic supergravity based on $J_{3}^{\mathbb{H}}$; in fact, these theories share the very same bosonic sector
\protect\cite{twin-Refs}.} \label{tab:D=5}
\end{table}


\begin{table}[H]
\begin{center}
\begin{tabular}{|c|ccc|ccc|ccc|}
\hline
$
\begin{array}{c}
\\
D=3:G_{\mathcal{N}}/H_{\mathcal{N}} \\
~
\end{array}
$ &  & $\mathbf{R}=\mathbf{Adj}$ &  &  & $d$ &  &  & $d_{g}$ &  \\
\hline\hline
$
\begin{array}{c}
\\
\mathcal{N}=5:\frac{USp(4,2n)}{USp(4)\times USp\left( 2n\right) },~~n\in \mathbb{N} \\
~
\end{array}
$ &  & $
\begin{array}{l}
\frac{\left( \mathbf{2n+5}\right) \left( \mathbf{2n+4}\right) }{\mathbf{2}}
\\
=\left( 2,0,...,0\right)
\end{array}
$ &  &  & $\begin{array}{l} 8\ (n>2) \\ 6\ (n=2)\\ 4\ (n=1)  \end{array}$ &  &  &
$\begin{array}{l} 6\ (n>2) \\ 4\ (n=2)\\ 2\ (n=1)  \end{array}$ &  \\ \hline
$
\begin{array}{c}
\\
\mathcal{N}=6:\frac{SU(4,n)}{SU(4)\times SU\left( n\right) \times U\left(1\right) },~~n\in \mathbb{N} \\
~
\end{array}
$ &  & $\begin{array}{l} \left( \mathbf{n+4}\right) ^{2}\mathbf{-1}\\ =\left( 1,0,....0,1\right)\end{array}
$ &  &  & $\begin{array}{l} 16\ (n>4) \\ 14\ (n=4)\\ 12\ (n=3)\\ 8\ (n=2)\\ 4\ (n=1)  \end{array}$ &  &  &
$\begin{array}{l} 14\ (n>4) \\ 12\ (n=4)\\ 10\ (n=3)\\ 6\ (n=2)\\ 2\ (n=1)  \end{array}$ &  \\ \hline
$
\begin{array}{c}
\\
\mathcal{N}=8:\frac{SO(8,n+2)}{SO(8)\times SO\left( n+2\right) }, \\
~ \\
n\in \mathbb{N}\cup \left\{ 0,-1\right\} ~~\left( \mathbb{R}\oplus \mathbf{\Gamma }_{n-1,5}\right) \\
~
\end{array}
$ &  & $
\begin{array}{l}
\frac{\left( \mathbf{n+10}\right) \left( \mathbf{n+9}\right) }{\mathbf{2}}
\\
=\left( 0,1,...,0\right)
\end{array}
$ &  &  &
$\begin{array}{l} 30\ (n>6) \\ 26\ (n=6,5)\\ 22\ (n=4)\\ 18\ (n=3)\\ 14\ (n=2)\\ 10\ (n=1)\\ 6\ (n=0)\\ 2\ (n=-1)  \end{array}$ &  &  &
$\begin{array}{l} 28\ (n>6) \\ 24\ (n=6,5)\\ 20\ (n=4)\\ 16\ (n=3)\\ 12\ (n=2)\\ 8\ (n=1)\\ 4\ (n=0)\\ 0\ (n=-1)  \end{array}$ &  \\ \hline
$
\begin{array}{c}
\\
\mathcal{N}=9:\frac{F_{4\left( -20\right) }}{SO\left( 9\right) } \\
~
\end{array}
$ &  & $\mathbf{52}=\left( 1000\right) $ &  &  & $4$ &  &  & $2$ &  \\ \hline
$
\begin{array}{c}
\\
\mathcal{N}=10:\frac{E_{6\left( -14\right) }}{SO(10)\times SO\left( 2\right)}~~\left( M_{1,2}\left( \mathbb{O}\right) \right) \\
~
\end{array}
$ &  & $\mathbf{78}=\left( 000001\right) $ &  &  & $8$ &  &  & $6$ &  \\
\hline
$
\begin{array}{c}
\\
\mathcal{N}=4\Leftrightarrow \mathcal{N}=12:\frac{E_{7\left( -5\right) }}{SO(12)\times SU\left( 2\right) }~~\left( J_{3}^{\mathbb{H}}\right) \\
~
\end{array}
$ &  & $\mathbf{133}=\left( 1000000\right) $ &  &  & $22$ &  &  & $20$ &  \\
\hline
$
\begin{array}{c}
\\
\mathcal{N}=16:\frac{E_{8\left( 8\right) }}{SO\left( 16\right) }~~\left(J_{3}^{\mathbb{O}_{s}}\right) \\
~
\end{array}
$ &  & $\mathbf{248}=\left( 00000010\right) $ &  &  & $58$ &  &  & $56$ &
\\ \hline
\end{tabular}
\end{center}
\caption{Scalar manifolds of $\mathcal{N}\geq 5$, $D=3$ supergravity theories. Note that the scalar manifold of the $\mathcal{N}=12$ theory
coincides with the one of $\mathcal{N}=4$ magic supergravity based on $J_{3}^{\mathbb{H}}$; in fact, these theories share the very same bosonic
sector \protect\cite{twin-Refs}. Note that, as for the $\mathcal{N}=4$ models with symmetric scalar manifold reported in Table~\protect\ref{tab:HKG-D=3}, the relevant $G$-irrep. is the adjoint.}
\label{tab:D=3}
\end{table}

\enlargethispage{2.5\baselineskip}
\vspace*{-2.8em}

\begin{table}[H]
\begin{center}
\begin{tabular}{|c|c|c|c|c||c|c|c|}
\hline
&  &  &  &  &  &  &  \\[-1.1em]
& $\mathbb{R}$ & $\mathbb{C}$ & $\mathbb{H}$ & $\mathbb{O}$ & $D$ & $d$ & $ d_g$ \\ \hline
&  &  &  &  &  &  &  \\
Str$_0\left(J_2^\mathbb{A}\right)$ & $SO(1,2)$ & $SO(1,3) \times SO(2)$ & $SO(1,5) \times SO(3)$ & $SO(1,9)$ & 6 & 2 & 0 \\[1em] \hline
&  &  &  &  &  &  &  \\
Str$_0\left(J_3^\mathbb{A}\right)$ & $SL(3,\mathbb{R})$ & $SL(3,\mathbb{C})$ & $SU^*(6)$ & $E_{6(-26)}$ & 5 & 4 & 2 \\[1em] \hline
&  &  &  &  &  &  &  \\
Conf$\left(J_3^\mathbb{A}\right)$ & $Sp(6,\mathbb{R})$ & $SU(3,3)$ & $SO^*(12)$ & $E_{7(-25)}$ & 4 & 9 & 8 \\[1em] \hline
&  &  &  &  &  &  &  \\
QConf$\left(J_3^\mathbb{A}\right)$ & $F_{4(4)}$ & $E_{6(2)}$ & $E_{7(-5)}$ & $E_{8(-24)}$ & 3 & 22 & 20 \\[1em] \hline
\end{tabular}
\end{center}
\caption{Values of degrees $d$ and $d_{g}$ of Iwasawa polynomials for magic Maxwell-Einstein supergravities in $D=3,4,5,6$ Lorentzian space-time
dimensions. They only depend on $D$, and not on $q:=\dim _{\mathbb{R}} \mathbb{A}$; this seems consistent with the fact that these homogeneous (symmetric) manifolds are in the same Tits-Satake universality class (see \textit{e.g.} \protect\cite{TS}).}
\label{tab:summary-1}
\end{table}

\begin{table}[H]
\begin{center}
\begin{tabular}{|c|ccc|ccc|ccc|}
\hline
$
\begin{array}{c}
\\
D=3:G_{3,\mathbb{A}_{s}}/mcs\left( G_{3,\mathbb{A}_{s}}\right) \\
~
\end{array}
$ &  & $
\begin{array}{c}
\\
\begin{array}{l}
G_{3,\mathbb{A}_{s}}\supset _{s}^{\max }G_{4,\mathbb{A}_{s}}\times SL(2,\mathbb{R}) \\
~ \\
\mathbf{R}_{3}\equiv \mathbf{Adj}\left( G_{3,\mathbb{A}_{s}}\right) \\
=\left( \mathbf{Adj}\left( G_{4,\mathbb{A}_{s}}\right) ,\mathbf{1}\right) \\
+\left( \mathbf{1},\mathbf{3}\right) +\left( \mathbf{R}_{4},\mathbf{2}\right)
\\
\end{array}
\\
~
\end{array}
$ &  &  & $d=d_{g}+2$ &  &  & $d_{g}=\dim _{\mathbb{R}}\mathbf{R}_{4}$ &  \\
\hline\hline
$
\begin{array}{c}
~ \\
\mathcal{N}=16:\frac{E_{8(8)}}{SO(16)} \\
~ \\
(J_{3}^{\mathbb{O}_{s}},~q=8)
\end{array}
$ &  & $
\begin{array}{l}
E_{8(8)}\supset _{s}^{\max }E_{7(7)}\times SL(2,\mathbb{R}) \\
~ \\
\mathbf{248}=\left( \mathbf{133},\mathbf{1}\right) +\left( \mathbf{1},\mathbf{3}\right) +\left( \mathbf{56},\mathbf{2}\right)
\end{array}
$ &  &  & $58$ &  &  & $56$ &  \\ \hline
$
\begin{array}{c}
~ \\
\mathcal{N}=0:\frac{E_{7(7)}}{SU(8)} \\
~ \\
(J_{3}^{\mathbb{H}_{s}},~q=4)
\end{array}
$ &  & $
\begin{array}{l}
E_{7(7)}\supset SO(6,6)\times SL(2,\mathbb{R}) \\
~ \\
\mathbf{133}=\left( \mathbf{66},\mathbf{1}\right) +\left( \mathbf{1},\mathbf{3}\right) +\left( \mathbf{32},\mathbf{2}\right)
\end{array}
$ &  &  & $34$ &  &  & $32$ &  \\ \hline
$%
\begin{array}{c}
~ \\
\mathcal{N}=0:\frac{E_{6(6)}}{USp(8)} \\
~ \\
(J_{3}^{\mathbb{C}_{s}},~q=2)
\end{array}
$ &  & $
\begin{array}{l}
E_{6(6)}\supset _{s}^{\max }SL(6,\mathbb{R})\times SL(2,\mathbb{R}) \\
~ \\
\mathbf{78}=\left( \mathbf{35},\mathbf{1}\right) +\left( \mathbf{1},\mathbf{3}\right) +\left( \mathbf{20},\mathbf{2}\right)
\end{array}
$ &  &  & $22$ &  &  & $20$ &  \\ \hline
\end{tabular}
\end{center}
\caption{Scalar manifolds of $D=3$ Maxwell-Einstein gravity theories based on rank-$3$ simple Jordan algebras $J_{3}^{\mathbb{A}_{s}}$on \textit{split}
algebras $\mathbb{A}_{s}=\mathbb{O}_{s}$, $\mathbb{H}_{s}$ and $\mathbb{C}_{s}$. The $D=3$ and $D=4$ $U$-duality groups are the quasiconformal resp.
conformal groups of $J_{3}^{\mathbb{A}_{s}}$ : $G_{3,\mathbb{A}_{s}}=QConf\left( J_{3}^{\mathbb{A}_{s}}\right) $,
$G_{4,\mathbb{A}_{s}}=Conf\left( J_{3}^{\mathbb{A}_{s}}\right) $. They occur in the fourth resp. third row of the symmetric Magic Square
$\mathcal{L}_{3}\left( \mathbb{A}_{s},\mathbb{B}_{s}\right) $ \protect\cite{BS,MS-1}. In $D=3$, only the
theory based on split octonions $\mathbb{O}_{s}$ can be regarded as the scalar sector of a locally supersymmetric theory (namely with maximal
supersymmetry in $D=3$ : $\mathcal{N}=16$), while the other theories are non-supersymmetric \protect\cite{BGM}. Interestingly, the degree $d$ of the
corresponding Iwasawa polynomial and the metric degree $d_{g}$ of $G_{3,\mathbb{A}_{s}}/mcs\left( G_{3,\mathbb{A}_{s}}\right) $ are related to the
dimension of the real symplectic $G_{4,\mathbb{A}_{s}}$-irrep. $\mathbf{R}_{4}$, in which the $2$-form Abelian field strengths of the corresponding
$D=4$ theory sit.}
\label{tab:DD=3}
\end{table}

\begin{table}[H]
\begin{center}
\begin{tabular}{|c|ccc|ccc|ccc|}
\hline
$
\begin{array}{c}
\\
D=4:G_{4,\mathbb{A}_{s}}/mcs\left( G_{4,\mathbb{A}_{s}}\right) \\
~
\end{array}
$ &  & $
\begin{array}{c}
\\
\begin{array}{l}
G_{4,\mathbb{A}_{s}}\supset _{s}^{\max }G_{5,\mathbb{A}_{s}}\times SO(1,1)
\\
~ \\
\mathbf{R}_{4}=\left( \mathbf{R}_{5}\right) _{1}+\left( \mathbf{R}_{5}^{\prime }\right) _{-1} \\
+\mathbf{1}_{3}+\mathbf{1}_{-3}
\end{array}
\\
~
\end{array}
$ &  &  & $d=\dim _{\mathbb{R}}\mathbf{R}_{5}$ &  &  & $d_{g}=4q$ &  \\
\hline\hline
$
\begin{array}{c}
~ \\
\mathcal{N}=8:\frac{E_{7(7)}}{SU(8)} \\
~ \\
(J_{3}^{\mathbb{O}_{s}},~q=8)
\end{array}
$ &  & $
\begin{array}{l}
E_{7(7)}\supset _{s}^{\max }E_{6(6)}\times SO(1,1) \\
~ \\
\mathbf{56}=\mathbf{27}_{1}+\mathbf{27}_{-1}^{\prime }+\mathbf{1}_{3}+\mathbf{1}_{-3}
\end{array}
$ &  &  & $27$ &  &  & $32$ &  \\ \hline
$
\begin{array}{c}
~ \\
\mathcal{N}=0:\frac{SO(6,6)}{SO(6)\times SO(6)} \\
~ \\
(J_{3}^{\mathbb{H}_{s}},~q=4)
\end{array}
$ &  & $
\begin{array}{l}
SO(6,6)\supset SL(6,\mathbb{R})\times SO(1,1) \\
~ \\
\mathbf{32}=\mathbf{15}_{1}+\mathbf{15}_{-1}^{\prime }+\mathbf{1}_{3}+\mathbf{1}_{-3}
\end{array}
$ &  &  & $15$ &  &  & $16$ &  \\ \hline
$
\begin{array}{c}
~ \\
\mathcal{N}=0:\frac{SL(6,\mathbb{R})}{SO(6)} \\
~ \\
(J_{3}^{\mathbb{C}_{s}},~q=2)
\end{array}
$ &  & $
\begin{array}{l}
SL(6,\mathbb{R})\supset _{s}^{\max } \\
SL(3,\mathbb{R})\times SL(3,\mathbb{R})\times SO(1,1) \\
~ \\
\mathbf{20}=\left( \mathbf{3}^{\prime },\mathbf{3}\right) _{1}+\left(\mathbf{3},\mathbf{3}^{\prime }\right) _{-1} \\
+\left( \mathbf{1,1}\right) _{3}+\left( \mathbf{1,1}\right) _{-3}
\end{array}
$ &  &  & $9$ &  &  & $8$ &  \\ \hline
\end{tabular}
\end{center}
\caption{Scalar manifolds of $D=4$ Maxwell-Einstein gravity theories based on rank-$3$ simple Jordan algebras $J_{3}^{\mathbb{A}_{s}}$on \textit{split}
algebras $\mathbb{A}_{s}=\mathbb{O}_{s}$ ($q=8$), $\mathbb{H}_{s}$ ($q=4$) and $\mathbb{C}_{s}$ ($q=2$). The $D=4$ and $D=5$ $U$-duality groups are the
conformal resp. reduced structure groups of $J_{3}^{\mathbb{A}_{s}}$ : $G_{4,\mathbb{A}_{s}}=Conf\left( J_{3}^{\mathbb{A}_{s}}\right) $, $G_{5,\mathbb{A}
_{s}}=Str_{0}\left( J_{3}^{\mathbb{A}_{s}}\right) $. They occur in the third resp.second row of the symmetric Magic Square
$\mathcal{L}_{3}\left( \mathbb{A}_{s},\mathbb{B}_{s}\right) $ \protect\cite{BS,MS-1}. In $D=4$, only the
theory based on split octonions $\mathbb{O}_{s}$ can be regarded as the scalar sector of a locally supersymmetric theory (namely with maximal
supersymmetry in $D=4$ : $\mathcal{N}=8$), while the other theories are non-supersymmetric \protect\cite{BGM}. Interestingly, the degree $d$ of the
corresponding Iwasawa polynomial is related to the dimension of the real symplectic $G_{5,\mathbb{A}_{s}}$-irrep. $\mathbf{R}_{5}$, in which the $2$-form
Abelian field strengths of the corresponding $D=5$ theory sit. On the other hand, the metric degree $d_{g}$ of $G_{4,\mathbb{A}_{s}}/mcs\left(
G_{4,\mathbb{A}_{s}}\right) $ is nothing but $4q$.}
\label{tab:DD=4}
\end{table}

\begin{table}[H]
\begin{center}
\begin{tabular}{|c|ccc|ccc|ccc|}
\hline
$
\begin{array}{c}
\\
D=5:G_{5,\mathbb{A}_{s}}/mcs\left( G_{5,\mathbb{A}_{s}}\right) \\
~
\end{array}
$ &  & $
\begin{array}{c}
\\
\begin{array}{l}
G_{5,\mathbb{A}_{s}}\supset _{s}^{\max }G_{6,\mathbb{A}_{s}}\times SO(1,1)
\\
~ \\
\mathbf{R}_{5}=\left( \mathbf{R}_{6}\right) _{1}+\left( \widehat{\mathbf{R}}_{6}\right) _{-2}+\mathbf{1}_{4}
\end{array}
\\
~
\end{array}
$ &  &  & $d=\dim _{\mathbb{R}}\mathbf{R}_{6}$ &  &  & $d_{g}$ &  \\
\hline\hline
$
\begin{array}{c}
~ \\
\mathcal{N}=8:\frac{E_{6(6)}}{USp(8)} \\
~ \\
(J_{3}^{\mathbb{O}_{s}},~q=8)
\end{array}
$ &  & $
\begin{array}{l}
E_{6(6)}\supset _{s}^{\max }SO(5,5)\times SO(1,1) \\
~ \\
\mathbf{27}=\mathbf{16}_{1}+\mathbf{10}_{-2}+\mathbf{1}_{4}
\end{array}
$ &  &  & $16$ &  &  & $20$ &  \\ \hline
$
\begin{array}{c}
~ \\
\mathcal{N}=0:\frac{SL(6,\mathbb{R})}{SO(6)} \\
~ \\
(J_{3}^{\mathbb{H}_{s}},~q=4)
\end{array}
$ &  & $
\begin{array}{l}
SL(6,\mathbb{R})\supset SL(4,\mathbb{R})\times SL(2,\mathbb{R})\times SO(1,1)
\\
~ \\
\mathbf{15}=\left( \mathbf{4,2}\right) _{1}+\left( \mathbf{6,1}\right)_{-2}+\left( \mathbf{1,1}\right) _{4}
\end{array}
$ &  &  & $8$ &  &  & $8$ &  \\ \hline
$
\begin{array}{c}
~ \\
\mathcal{N}=0:\frac{SL(3,\mathbb{R})}{SO(3)}\times \frac{SL(3,\mathbb{R})}{SO(3)} \\
~ \\
(J_{3}^{\mathbb{C}_{s}},~q=2)
\end{array}
$ &  & $
\begin{array}{l}
SL(3,\mathbb{R})\times SL(3,\mathbb{R}) \\
~\supset _{s}^{\max }SL(2,\mathbb{R})\times SL(2,\mathbb{R}) \\
\times SO(1,1)\times SO(1,1) \\
\\
\left( \mathbf{3}^{\prime },\mathbf{3}\right) =\left( \mathbf{2},\mathbf{2}\right) _{-1,1}+\left( \mathbf{2},\mathbf{1}\right) _{-1,-2} \\
+\left( \mathbf{1},\mathbf{2}\right) _{2,1}+\left( \mathbf{1},\mathbf{1}
\right) _{2,-2}
\end{array}
$ &  &  & $4$ &  &  & $2$ &  \\ \hline
\end{tabular}
\end{center}
\caption{Scalar manifolds of $D=5$ Maxwell-Einstein gravity theories based on rank-$3$ simple Jordan algebras $J_{3}^{\mathbb{A}_{s}}$on \textit{split}
algebras $\mathbb{A}_{s}=\mathbb{O}_{s}$, $\mathbb{H}_{s}$ and $\mathbb{C}_{s}$. The $D=5$ and $D=6$ $U$-duality groups are the structure groups of
$J_{3}^{\mathbb{A}_{s}}$ resp. $J_{2}^{\mathbb{A}_{s}}\sim \mathbf{\Gamma }_{q/2+1,q/2+1}$: $G_{5,\mathbb{A}_{s}}=Str_{0}
\left( J_{3}^{\mathbb{A}_{s}}\right) $, and $G_{6,\mathbb{A}_{s}}=Str_{0}\left( J_{2}^{\mathbb{A}_{s}}\right) =SO\left( q/2+1,q/2+1\right) \times
\frac{Tri(\mathbb{A}_{s)}}{SO(\mathbb{A}_{s})}$. $G_{5,\mathbb{A}_{s}}$ occur in the second row of the symmetric Magic Square $\mathcal{L}_{3}\left(
\mathbb{A}_{s},\mathbb{B}_{s}\right) $ \protect\cite{BS,MS-1}, whereas $G_{6,\mathbb{A}_{s}}$ occur in the first row of Table~\protect\ref{tab:summary-1}.
$Tri(\mathbb{A}_{s})$ and $SO(\mathbb{A}_{s})$ respectively denote the triality and norm-preserving groups of $\mathbb{A}_{s}$, and the factor
$\frac{Tri(\mathbb{A}_{s})}{SO(\mathbb{A}_{s})}$ is non-trivial only for $q=4$ ($SL(2,\mathbb{R})$) and for $q=2$ ($SO(1,1)$). In $D=5$, only the theory based on
split octonions $\mathbb{O}_{s}$ can be regarded as the scalar sector of a locally supersymmetric theory (namely with maximal supersymmetry in $D=5$ :
$\mathcal{N}=8$), while the other theories are non-supersymmetric. Interestingly, the degree $d$ of the corresponding Iwasawa polynomial is
related to the dimension of the real symplectic $G_{6,\mathbb{A}_{s}}$-irrep. $\mathbf{R}_{6}$, in which the vector multiplets of the
corresponding (anomaly-free) $D=6$ theory sit.}
\label{tab:DD=5}
\end{table}

\section*{Acknowledgments}

AM and SF would like to thank Renata Kallosh for enlightening discussions. BLC and SF would like to thank Mario Trigiante for useful observations.

The work of BLC is supported in part by the European Commission under the FP7- PEOPLE-IRG-2008 Grant No. PIRG04-GA-2008-239412 ``\textit{String Theory and Noncommutative Ge\-ometry”}" (STRING). BLC would like to thank the CERN Theoretical Physics Division, the ERC Advanced Grant no. 226455 \textit{SUPERFIELDS} and the University of California Berkeley Center for Theoretical Physics, where part of this work was done, for kind hospitality and stimulating environment.

The work of SF is supported by the ERC Advanced Grant no. 226455 \textit{SUPERFIELDS}.

The work of AM is supported in part by the FWO - Vlaanderen, Project No.
G.0651.11, and in part by the Interuniversity Attraction Poles Programme initiated by the Belgian Science Policy (P7/37).

\begin{appendix}

\section{\label{App-Racah}On Racah-Casimir Polynomials in $\mathfrak{g}$}

A systematic study of primitive invariant tensors of a \textit{compact}, \textit{\ simple} Lie algebra $\mathfrak{g}$ has been the subject of a
number of monographies and papers along the years; here, we will give a concise \textit{r\'{e}sum\'{e}}, mainly based on
\cite{Azcarraga-et-al,Mountain,MacFarlane-Pfeiffer,Halperin}, to which we address the reader for further elucidation and a list of references.

The symmetric invariant tensors give rise to the so-called Racah-Casimir polynomials of $\mathfrak{g}$; on the other hand, the skewsymmetric
invariant tensors determine the non-trivial cocycles for the Lie algebra cohomology (see \textit{e.g.} \cite{CE}). By denoting the rank of
$\mathfrak{g}$ with $l$, it is well known \cite{1} that there are $l$ such invariant symmetric \textit{primitive} polynomials of order $\delta _{A}$ ($A=1,\dots
,l$), which determine $l$ independent \textit{primitive} Racah-Casimir polynomials $\left\{ \mathcal{C}_{\delta _{A}}\right\} _{A=1,...,l}$ of the
same order, as well as $l$ skewsymmetric invariant primitive tensors $\Omega^{(2\delta _{A}-1)}$ of order $2\delta _{A}-1$. The latter determine
the non-trivial cocycles for the Lie algebra cohomology, their order being related to the topological properties of the associated compact group
manifold $G$ which, from the point of view of the real homology, behaves as products of $l$ spheres $S^{(2\delta _{A}-1)}$ \cite{2}. Indeed, the
Poincar\'{e} polynomial of $G$ is of the form (see \textit{e.g.} \cite{Halperin})
\begin{eqnarray}
f_{G}\left( t\right) &=&\prod\limits_{A=1}^{l}\left(1+t^{2\delta_{A}-1}\right) ; \\
\sum_{A=1}^{l}\left( 2\delta _{A}-1\right) &=&\text{dim}\left( G\right) .
\end{eqnarray}

Remarkably, the so-called \textit{principal} $SU(2)_{P}$ \cite{Kostant-1} is (generally non-symmetrically) embedded in $G$ such that (\ref{embb-1})
holds, with $\delta _{A}=j_{A}+1$. All simple Lie groups admits an embedded \textit{principal} $SU(2)_{P}$, whose embedding is always maximal and
non-symmetric. Exceptions are provided by $SU(3)$, which embeds $SO(3)\simeq SU(2)_{P}$ symmetrically, and by $E_{6}$, which embeds $SU(2)_{P}$ only
\textit{next-to-maximally}, \textit{i.e.} through the $2$-step chain of maximal embeddings:
\begin{equation}
E_{6}\supset _{s}F_{4}\supset SU(2)_{P},
\end{equation}
where the subscript \textquotedblleft $s$" in the first embedding denotes that it is symmetric.

The $\mathcal{C}_{\delta _{A}}$'s are particular homogeneous polynomials in $U\left( \mathfrak{g}\right) $, the universal enveloping algebra of
$\mathfrak{g}$ : indeed, they generate the \textit{center} $U\left( \mathfrak{g}\right) ^{\mathfrak{g}}$ of $U\left( \mathfrak{g}\right) $ itself.
Remarkably, $\left\{ \mathcal{C}_{\delta _{A}}\right\} _{A=1,...,l}$ constitutes a complete, \textquotedblleft minimal-degree", \textit{finitely
generating} basis of the ring of invariant polynomials in $\mathfrak{g}$ (see \textit{e.g.} \cite{Kac-80}). Their degrees, \textit{i.e.} the numbers
$\delta _{A}$, are known for all simple $\mathfrak{g}$, and they are reported in Table~\ref{tab:fourierFg3}; they can be computed by diagonalizing the Coxeter element, the
product of simple Weyl reflections (see \textit{e.g.} \cite{Humphreys-Dixmier}). The relation between $\delta _{A}$'s and the Cartan
matrix of $\mathfrak{g}$ has been determined in \cite{mi-Cartan}, and recently reviewed in many cases in App. A of \cite{SW-revisited}. It should
also be recalled that a neat derivation of the Betti numbers $2\delta _{A}-1=2j_{A}+1$ of semisimple Lie groups is presented in
\cite{Freudenthal-deVries}.
\begin{table}[tbph]
\begin{equation*}
\begin{array}{l|l|l||l|l|l|}
\mathfrak{g} & \delta _{A} & C_G & \mathfrak{g} & \delta _{A} & C_G \\ \hline
A_{n-1} & 2,3,\ldots ,n & n & B_{n} & 2,4,\ldots ,2n & 2n \\
D_{n} & 2,4,\ldots ,2n-2;n & 2n-2 & C_{n} & 2,4,\ldots ,2n & 2n \\
\mathfrak{e}_{6} & 2,5,6,8,9,12 & 12 & \mathfrak{f}_{4} & 2,6,8,12 & 12 \\
\mathfrak{e}_{7} & 2,6,8,10,12,14,18 & 18 & \mathfrak{g}_{2} & 2,6 & 6 \\
\mathfrak{e}_{8} & 2,8,12,14,18,20,24,30 & 30 &  &  &
\end{array}
\end{equation*}
\caption{The orders $\protect\delta _{A}$ of the Casimir invariant polynomials $\left\{ \mathcal{C}_{\protect\delta _{A}}\right\} _{A=1,...,l}$,
and the Coxeter number $C_{G}=\max \left\{ \protect\delta _{A}\right\} $ for each simple Lie algebra $\mathfrak{g}$. Recall $A_{n-1}=\mathfrak{su}(n)$,
$B_{n}=\mathfrak{so}(2n+1)$, $C_{n}=\mathfrak{sp}(2n)$, $D_{n}=\mathfrak{so}(2n)$. Simply-laced and non-simply-laced $\mathfrak{g}$'s are respectively
listed on the left and right hand side.}
\end{table}

\begin{table}[tbph]
\begin{equation*}
\begin{array}{l|l||l|l}
\mathfrak{g} & \mathbf{Adj}\left( G\right) =\sum_{A=1}^{l}\left( \mathbf{2j}_{A}+\mathbf{1}\right) & \mathfrak{g} & \mathbf{Adj}\left( G\right)
=\sum_{A=1}^{l}\left( \mathbf{2j}_{A}+\mathbf{1}\right)  \\ \hline A_{n-1} & \mathbf{n}^{2}\mathbf{-1}=\mathbf{3}+\mathbf{5}+...(\mathbf{2n-1})
& B_{n} & \mathbf{n}(\mathbf{2n+1})=\mathbf{3}+\mathbf{7}+...(\mathbf{4n-1})
\\
D_{n} & \mathbf{n}(\mathbf{2n-1})=\mathbf{3}+\mathbf{7}+...+\left( \mathbf{4n-5}\right) \mathbf{+}\mathbf{\left( 2n-1\right) } & C_{n} & \mathbf{n}
(\mathbf{2n+1})=\mathbf{3}+\mathbf{7}+...(\mathbf{4n-1}) \\
\mathfrak{e}_{6} & \mathbf{78}=\mathbf{3}+\mathbf{9}+\mathbf{11}+\mathbf{15}+\mathbf{17}+\mathbf{23} & \mathfrak{f}_{4} & \mathbf{52}=\mathbf{3}+
\mathbf{11}+\mathbf{15}+\mathbf{23} \\
\mathfrak{e}_{7} & \mathbf{133}=\mathbf{3}+\mathbf{11}+\mathbf{15}+\mathbf{19}+\mathbf{23}+\mathbf{27}+\mathbf{35} & \mathfrak{g}_{2} & \mathbf{14}=
\mathbf{3}+\mathbf{11} \\
\mathfrak{e}_{8} & \mathbf{248}=\mathbf{3}+\mathbf{15}+\mathbf{23}+\mathbf{27}+\mathbf{35}+\mathbf{39}+\mathbf{47}+\mathbf{59} &  &
\end{array}
\end{equation*}
\caption{Decomposition of the adjoint irrep. (\protect\ref{embb-1}) under the maximal embedding (\protect\ref{pre-emb-1}) of $\mathfrak{su}(2)_{P}$
into $\mathfrak{g}$ \protect\cite{Kostant-1}.}
\end{table}

\section{\label{app:semispin}Semispin Groups}

Among the groups of type $D_{2m}$ there are four interesting compact forms which we will shortly describe here \cite{McInnes}.\newline
The first one is the spin group $Spin(4m)$ that is a double covering of $SO(4m)$. It is realized as usual as a multiplicative subgroup of the even
Clifford algebra associated to $\mathbb{R}^{4m}$ with the standard Euclidean product. If $\{ e_1, \ldots, e_{4m}\}$ is the canonical basis of
$\mathbb{R}^{4m}$ (naturally embedded into the Clifford algebra) then, setting $e:=e_1\cdot\ldots\cdot e_{2m}$, one gets that $\pm e$ generate the center of
$Spin(4m)$. Since $e^2=1$ the center is thus a group $\mathbb{Z}_2\times \mathbb{Z}_2$
\begin{eqnarray}
Z=\{1, e \} \times \{1, -e\}\equiv Z_1\times Z_2.
\end{eqnarray}
$Z$ contains a third $\mathbb{Z}_2$ subgroup, the diagonal subgroup $Z_3={1,-1}$. Thus, one can construct three quotient groups $S(4m)_i:=Spin(4m)/Z_i$.
For $m=2$ the triality provides an isomorphism among the three groups so that $S(8)_i\simeq SO(8)$. But for $m\geq3$ there is no triality and only
$S(4m)_1 \simeq S(4m)_2$ are isomorphic. Thus we have two distinct quotients
\begin{eqnarray}
SO(4m)\simeq S(4m)_3, \qquad\qquad\ S_s(4m):=S(4m)_1.
\end{eqnarray}
The last one is called the semispin group.\newline
Finally, one can consider the group $PSO(4m)=SO(4m)/{\pm I_{4m}}$, where $\pm I_{4m}$ is the image of the center of $Spin(4m)$ in the projection
$Spin(4m)\mapsto SO(4m)$.


\section{\label{app:Inverse-Cartan}Inverse Cartan Matrices}

\label{app:cartan} For sake of completeness, here we list the inverse Cartan
matrices for all simple Lie groups.
\begin{eqnarray}
\pmb C^{-1}_{A_n}=\frac 1{n+1}
\begin{pmatrix}
n & n-1 & n-2 & \ldots & 3 & 2 & 1 \\
n-1 & 2(n-1) & 2(n-2) & \ldots & 3\cdot 2 & 2\cdot 2 & 2 \\
n-2 & 2(n-2) & 3(n-2) & \ldots & 3\cdot 3 & 3\cdot 2 & 3 \\
\ldots & \ldots & \ldots & \ldots & \ldots & \ldots & \ldots \\
3 & 2\cdot 3 & 3\cdot 3 & \ldots & (n-2)3 & (n-2)2 & n-2 \\
2 & 2\cdot 2 & 3\cdot 2 & \ldots & (n-2)2 & (n-1)2 & n-1 \\
1 & 2 & 3 & \ldots & n-2 & n-1 & n
\end{pmatrix}
,
\end{eqnarray}
\begin{eqnarray}
\pmb C^{-1}_{B_n}=
\begin{pmatrix}
1 & 1 & 1 & \ldots & 1 & 1 & 1 \\
1 & 2 & 2 & \ldots & 2 & 2 & 2 \\
1 & 2 & 3 & \ldots & 3 & 3 & 3 \\
\ldots & \ldots & \ldots & \ldots & \ldots & \ldots & \ldots \\
1 & 2 & 3 & \ldots & n-2 & n-2 & n-2 \\
1 & 2 & 3 & \ldots & n-2 & n-1 & n-1 \\
\frac 12 & \frac 22 & \frac 32 & \ldots & \frac {n-2}2 & \frac {n-1}2 &
\frac n2
\end{pmatrix}
,
\end{eqnarray}
\begin{eqnarray}
\pmb C^{-1}_{C_n}=
\begin{pmatrix}
1 & 1 & 1 & \ldots & 1 & 1 & 1/2 \\
1 & 2 & 2 & \ldots & 2 & 2 & 2/2 \\
1 & 2 & 3 & \ldots & 3 & 3 & 3/2 \\
\ldots & \ldots & \ldots & \ldots & \ldots & \ldots & \ldots \\
1 & 2 & 3 & \ldots & n-2 & n-2 & (n-2)/2 \\
1 & 2 & 3 & \ldots & n-2 & n-1 & (n-1)/2 \\
1 & 2 & 3 & \ldots & n-2 & n-1 & n/2
\end{pmatrix}
,
\end{eqnarray}
\begin{eqnarray}
\pmb C^{-1}_{D_n}=
\begin{pmatrix}
1 & 1 & 1 & \ldots & 1 & 1/2 & 1/2 \\
1 & 2 & 2 & \ldots & 2 & 2/2 & 2/2 \\
1 & 2 & 3 & \ldots & 3 & 3/2 & 3/2 \\
\ldots & \ldots & \ldots & \ldots & \ldots & \ldots & \ldots \\
1 & 2 & 3 & \ldots & n-2 & (n-2)/2 & (n-2)/2 \\
1/2 & 2/2 & 3/2 & \ldots & (n-2)/2 & n/4 & (n-2)/4 \\
1/2 & 2/2 & 3/2 & \ldots & (n-2)/2 & (n-2)/4 & n/4
\end{pmatrix}
,
\end{eqnarray}
\begin{eqnarray}
\pmb C^{-1}_{G_2}=
\begin{pmatrix}
2 & 3 \\
1 & 2
\end{pmatrix}
,
\end{eqnarray}
\begin{eqnarray}
\pmb C^{-1}_{F_4}=
\begin{pmatrix}
2 & 3 & 4 & 2 \\
3 & 6 & 8 & 4 \\
2 & 4 & 6 & 3 \\
1 & 2 & 3 & 2
\end{pmatrix}
,
\end{eqnarray}
\begin{eqnarray}
\pmb C^{-1}_{E_6}=
\begin{pmatrix}
4/3 & 5/3 & 2 & 4/3 & 2/3 & 1 \\
5/3 & 10/3 & 4 & 8/3 & 4/3 & 2 \\
2 & 4 & 6 & 4 & 2 & 3 \\
4/3 & 8/3 & 4 & 10/3 & 5/3 & 2 \\
2/3 & 4/3 & 2 & 5/3 & 4/3 & 1 \\
1 & 2 & 3 & 2 & 1 & 2
\end{pmatrix}
,
\end{eqnarray}
\begin{eqnarray}
\pmb C^{-1}_{E_7}=
\begin{pmatrix}
2 & 3 & 4 & 3 & 2 & 1 & 2 \\
3 & 6 & 8 & 6 & 4 & 2 & 4 \\
4 & 8 & 12 & 9 & 6 & 3 & 6 \\
3 & 6 & 9 & 15/2 & 5 & 5/2 & 9/2 \\
2 & 4 & 6 & 5 & 4 & 2 & 3 \\
1 & 2 & 3 & 5/2 & 2 & 3/2 & 3/2 \\
2 & 4 & 6 & 9/2 & 3 & 3/2 & 7/2
\end{pmatrix}
,
\end{eqnarray}
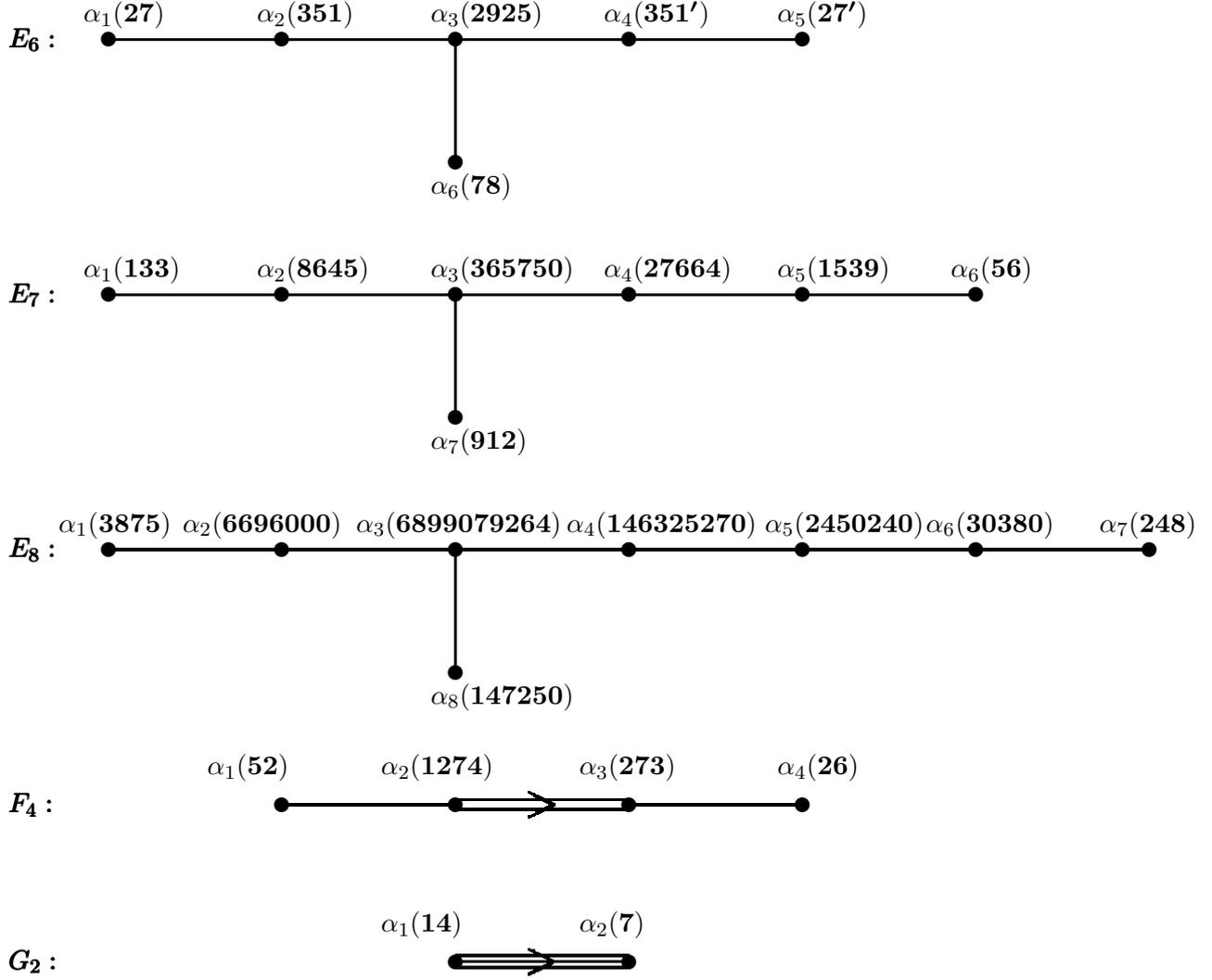
\begin{eqnarray}
\pmb C^{-1}_{E_8}=
\begin{pmatrix}
4 & 7 & 10 & 8 & 6 & 4 & 2 & 5 \\
7 & 14 & 20 & 16 & 12 & 8 & 4 & 10 \\
10 & 20 & 30 & 24 & 18 & 12 & 6 & 15 \\
8 & 16 & 24 & 20 & 15 & 10 & 5 & 12 \\
6 & 12 & 18 & 15 & 12 & 8 & 4 & 9 \\
4 & 8 & 12 & 10 & 8 & 6 & 3 & 6 \\
2 & 4 & 6 & 5 & 4 & 3 & 2 & 3 \\
5 & 10 & 15 & 12 & 9 & 6 & 3 & 8
\end{pmatrix}
.
\end{eqnarray}
\section{\label{app:Dynkin-diagrams}Dynkin diagrams}
Here we show the Dynkin diagrams for all the simple groups, with the fundamental representations corresponding
to the simple weights associated to the simple roots.

\begin{figure}[h]
\label{figAn}\centering
\begin{picture}(200, 60)(-50,-20)
\put(-190,0){\makebox(0,0)[l]{$\pmb{A_n:}$}}
\put(-180,20){\makebox(0,0)[l]{{$\alpha_1$({\boldmath{$n+1$}})}}}
\put(-120,20){\makebox(0,0)[l]{{$\alpha_2$({\boldmath{$\begin{pmatrix} n+1 \\ 2 \end{pmatrix}$}})}}}
\put(-40,20){\makebox(0,0)[l]{{$\alpha_i$({\boldmath{$\begin{pmatrix} n+1 \\ i \end{pmatrix}$}})}}}
\put(30,20){\makebox(0,0)[l]{{$\alpha_{n-1}$({\boldmath{$\begin{pmatrix} n+1 \\ 2 \end{pmatrix}^*$}})}}}
\put(120,20){\makebox(0,0)[l]{{$\alpha_{n}$(\boldmath{$(n+1)^*$})}}}
\put(-150,0){\circle*{6}}
\put(-80,0){\circle*{6}}
\put(-10,0){\circle*{6}}
\put(60,0){\circle*{6}}
\put(130,0){\circle*{6}}
\thicklines
\dottedline{5}(-80,0)(60,0)
\drawline(-150,0)(-80,0)
\drawline(60,0)(130,0)
\end{picture}
\end{figure}
\begin{figure}[h]
\label{figBn}\centering
\begin{picture}(200, 90)(-50,-60)
\put(-190,0){\makebox(0,0)[l]{$\pmb{B_n:}$}}
\put(-180,20){\makebox(0,0)[l]{{$\alpha_1$({\boldmath{$2n+1$}})}}}
\put(-110,20){\makebox(0,0)[l]{{$\alpha_2$({\boldmath{$\begin{pmatrix} 2n+1 \\ 2 \end{pmatrix}$}})}}}
\put(-40,-20){\makebox(0,0)[l]{{$\alpha_i$({\boldmath{$\begin{pmatrix} 2n+1 \\ i \end{pmatrix}$}})}}}
\put(10,20){\makebox(0,0)[l]{{$\alpha_{n-2}$({\boldmath{$\begin{pmatrix} 2n+1 \\ n-2 \end{pmatrix}$}})}}}
\put(90,-20){\makebox(0,0)[l]{{$\alpha_{n-1}$(\boldmath{$\begin{pmatrix} 2n+1 \\ n-1 \end{pmatrix}$})}}}
\put(190,20){\makebox(0,0)[l]{{$\alpha_{n}$(\boldmath{$2^n$})}}}
\put(-150,0){\circle*{6}}
\put(-80,0){\circle*{6}}
\put(-10,0){\circle*{6}}
\put(60,0){\circle*{6}}
\put(130,0){\circle*{6}}
\put(200,0){\circle*{6}}
\thicklines
\dottedline{5}(-80,0)(60,0)
\drawline(-150,0)(-80,0)
\drawline(60,0)(130,0)
\drawline(130,-2)(200,-2)
\drawline(130,2)(200,2)
\drawline(170,0)(160,5)
\drawline(170,0)(160,-5)
\end{picture}
\end{figure}
\begin{figure}[h]
\label{figCn}\centering
\begin{picture}(200, 90)(-50,-80)
\put(-190,0){\makebox(0,0)[l]{$\pmb{C_n:}$}}
\put(-160,20){\makebox(0,0)[l]{{$\alpha_1$({\boldmath{$2n$}})}}}
\put(-110,20){\makebox(0,0)[l]{{$\alpha_2$({\boldmath{$\begin{pmatrix} 2n \\ 2 \end{pmatrix}-1$}})}}}
\put(-60,-20){\makebox(0,0)[l]{{$\alpha_i$({\boldmath{$\begin{pmatrix} 2n \\ i \end{pmatrix}-\begin{pmatrix} 2n \\ i-2 \end{pmatrix}$}})}}}
\put(-10,20){\makebox(0,0)[l]{{$\alpha_{n-2}$({\boldmath{$\begin{pmatrix} 2n \\ n-2 \end{pmatrix}-\begin{pmatrix} 2n \\ n-4 \end{pmatrix}$}})}}}
\put(70,-20){\makebox(0,0)[l]{{$\alpha_{n-1}$(\boldmath{$\begin{pmatrix} 2n \\ n-1 \end{pmatrix}-\begin{pmatrix} 2n \\ n-3 \end{pmatrix}$})}}}
\put(150,20){\makebox(0,0)[l]{{$\alpha_{n}$(\boldmath{$\begin{pmatrix} 2n \\ n \end{pmatrix}-\begin{pmatrix} 2n \\ n-2 \end{pmatrix}$})}}}
\put(-150,0){\circle*{6}}
\put(-80,0){\circle*{6}}
\put(-10,0){\circle*{6}}
\put(60,0){\circle*{6}}
\put(130,0){\circle*{6}}
\put(200,0){\circle*{6}}
\thicklines
\dottedline{5}(-80,0)(60,0)
\drawline(-150,0)(-80,0)
\drawline(60,0)(130,0)
\drawline(130,-2)(200,-2)
\drawline(130,2)(200,2)
\drawline(160,0)(170,5)
\drawline(160,0)(170,-5)
\end{picture}
\end{figure}
\vspace{-4.5em}
\begin{figure}[h]
\label{figDn}\centering
\begin{picture}(200, 80)(-50,-40)
\put(-190,0){\makebox(0,0)[l]{$\pmb{D_n:}$}}
\put(-160,20){\makebox(0,0)[l]{{$\alpha_1$({\boldmath{$2n$}})}}}
\put(-110,20){\makebox(0,0)[l]{{$\alpha_2$({\boldmath{$\begin{pmatrix} 2n \\ 2 \end{pmatrix}$}})}}}
\put(-40,20){\makebox(0,0)[l]{{$\alpha_i$({\boldmath{$\begin{pmatrix} 2n \\ i \end{pmatrix}$}})}}}
\put(20,20){\makebox(0,0)[l]{{$\alpha_{n-3}$({\boldmath{$\begin{pmatrix} 2n \\ n-3 \end{pmatrix}$}})}}}
\put(100,20){\makebox(0,0)[l]{{$\alpha_{n-2}$(\boldmath{$\begin{pmatrix} 2n \\ n-2 \end{pmatrix}$})}}}
\put(190,20){\makebox(0,0)[l]{{$\alpha_{n-1}$(\boldmath{$2^{n-1}$})}}}
\put(120,-60){\makebox(0,0)[l]{{$\alpha_n$(\boldmath{$2^{n-1\prime}$})}}}
\put(-150,0){\circle*{6}}
\put(-80,0){\circle*{6}}
\put(-10,0){\circle*{6}}
\put(60,0){\circle*{6}}
\put(130,0){\circle*{6}}
\put(200,0){\circle*{6}}
\put(130,-50){\circle*{6}}
\thicklines
\dottedline{5}(-80,0)(60,0)
\drawline(-150,0)(-80,0)
\drawline(60,0)(200,0)
\drawline(130,0)(130,-50)
\end{picture}
\end{figure}
\begin{figure}[h]
\label{figE6}\centering
\begin{picture}(200, 90)(-50,-70)
\put(-190,0){\makebox(0,0)[l]{$\pmb{E_6:}$}}
\put(-160,10){\makebox(0,0)[l]{{$\alpha_1$({\bf 27})}}}
\put(-90,10){\makebox(0,0)[l]{{$\alpha_2$({\bf 351})}}}
\put(-20,10){\makebox(0,0)[l]{{$\alpha_3$({\bf 2925})}}}
\put(50,10){\makebox(0,0)[l]{{$\alpha_4$({\boldmath{$351'$}})}}}
\put(120,10){\makebox(0,0)[l]{{$\alpha_5$({\boldmath{$27'$}})}}}
\put(-20,-60){\makebox(0,0)[l]{{$\alpha_6$({\bf 78})}}}
\put(-150,0){\circle*{6}} \put(-80,0){\circle*{6}}
\put(-10,0){\circle*{6}} \put(60,0){\circle*{6}}
\put(130,0){\circle*{6}} \put(-10,-50){\circle*{6}}
\thicklines \drawline(-150,0)(130,0) \drawline(-10,0)(-10,-50)
\end{picture}
\end{figure}
\begin{figure}[h]
\label{figE7}\centering
\begin{picture}(200, 90)(-50,-70)
\put(-190,0){\makebox(0,0)[l]{$\pmb{E_7:}$}}
\put(-160,10){\makebox(0,0)[l]{{$\alpha_1$({\bf 133})}}}
\put(-90,10){\makebox(0,0)[l]{{$\alpha_2$({\bf 8645})}}}
\put(-20,10){\makebox(0,0)[l]{{$\alpha_3$({\bf 365750})}}}
\put(50,10){\makebox(0,0)[l]{{$\alpha_4$({\bf 27664})}}}
\put(120,10){\makebox(0,0)[l]{{$\alpha_5$({\bf 1539})}}}
\put(190,10){\makebox(0,0)[l]{{$\alpha_6$({\bf 56})}}}
\put(-20,-60){\makebox(0,0)[l]{{$\alpha_7$({\bf 912})}}}
\put(-150,0){\circle*{6}}
\put(-80,0){\circle*{6}}
\put(-10,0){\circle*{6}}
\put(60,0){\circle*{6}}
\put(130,0){\circle*{6}}
\put(200,0){\circle*{6}}
\put(-10,-50){\circle*{6}}
\thicklines
\drawline(-150,0)(200,0)
\drawline(-10,0)(-10,-50)
\end{picture}
\end{figure}
\begin{figure}[h!]
\label{figE8}\centering
\begin{picture}(200, 90)(-50,-70)
\put(-190,0){\makebox(0,0)[l]{$\pmb{E_8:}$}}
\put(-170,10){\makebox(0,0)[l]{{$\alpha_1$({\bf 3875})}}}
\put(-120,10){\makebox(0,0)[l]{{$\alpha_2$({\bf 6696000})}}}
\put(-50,10){\makebox(0,0)[l]{{$\alpha_3$({\bf 6899079264})}}}
\put(35,10){\makebox(0,0)[l]{{$\alpha_4$({\bf 146325270})}}}
\put(115,10){\makebox(0,0)[l]{{$\alpha_5$({\bf 2450240})}}}
\put(180,10){\makebox(0,0)[l]{{$\alpha_6$({\bf 30380})}}}
\put(250,10){\makebox(0,0)[l]{{$\alpha_7$({\bf 248})}}}
\put(-20,-60){\makebox(0,0)[l]{{$\alpha_8$({\bf 147250})}}}
\put(-150,0){\circle*{6}} \put(-80,0){\circle*{6}}
\put(-10,0){\circle*{6}} \put(60,0){\circle*{6}}
\put(130,0){\circle*{6}} \put(200,0){\circle*{6}}
\put(270,0){\circle*{6}} \put(-10,-50){\circle*{6}}
\thicklines
\drawline(-150,0)(270,0)
\drawline(-10,0)(-10,-50)
\end{picture}
\end{figure}

\begin{figure}[h!]
\label{figF4}\centering
\begin{picture}(200, 90)(-50,-70)
\put(-190,0){\makebox(0,0)[l]{$\pmb{F_4:}$}}
\put(-110,15){\makebox(0,0)[l]{{$\alpha_1$({\boldmath{$52$}})}}}
\put(-40,15){\makebox(0,0)[l]{{$\alpha_2$({\boldmath{$1274$}})}}}
\put(40,15){\makebox(0,0)[l]{{$\alpha_{3}$({\boldmath{$273$}})}}}
\put(120,15){\makebox(0,0)[l]{{$\alpha_{4}$(\boldmath{$26$})}}}
\put(-80,0){\circle*{6}}
\put(-10,0){\circle*{6}}
\put(60,0){\circle*{6}}
\put(130,0){\circle*{6}}
\thicklines
\drawline(-80,0)(-10,0)
\drawline(60,0)(130,0)
\drawline(-10,-2)(60,-2)
\drawline(-10,2)(60,2)
\drawline(30,0)(20,5)
\drawline(30,0)(20,-5)
\end{picture}
\end{figure}
\begin{figure}[h!]
\label{figG2}\centering
\begin{picture}(200, 10)(-50,-30)
\put(-190,0){\makebox(0,0)[l]{$\pmb{G_2:}$}}
\put(-40,15){\makebox(0,0)[l]{{$\alpha_1$({\boldmath{$14$}})}}}
\put(40,15){\makebox(0,0)[l]{{$\alpha_{2}$({\boldmath{$7$}})}}}
\put(-10,0){\circle*{6}}
\put(60,0){\circle*{6}}
\thicklines
\drawline(-10,-2.5)(60,-2.5)
\drawline(-10,2.5)(60,2.5)
\drawline(-10,0)(60,0)
\drawline(30,0)(20,5)
\drawline(30,0)(20,-5)
\end{picture}
\end{figure}


\newpage

\section{\label{app:Satake-Type}The Satake Type Vectors}

\label{app:satake} Here we give a complete list of the vectors with entry $1$ if corresponding to a white dot of the Satake diagram and zero otherwise. In
the formulas $\bar e_i$, $i=1,\ldots, n$ indicates the canonical basis of $ \mathbb{R}^n$, where $n$ is the rank of the group $G_{nc}$ from which the
NISS is realized, indicated in parenthesis. For the meaning of the indices $ n $, $p$, $k$, refer to Table~\ref{tab:degrees}.
\begin{eqnarray}
&& \bar \varepsilon_{AI(n)}=\sum_{i=1}^n \bar e_i, \\
&& \bar \varepsilon_{AII(2k-1)}=\sum_{i=1}^{k-1} \bar e_{2i},\\
&& \bar \varepsilon_{AIII_a(2n-1)}=\sum_{i=1}^{2n-1} \bar e_{i},
\end{eqnarray}
\begin{eqnarray}
&& \bar \varepsilon_{AIII_b(2p-1)}=\sum_{i=1}^{p} (\bar e_{i} +\bar e_{n-i}),\\
&& \bar \varepsilon_{AIV(n)}=\bar e_{1}+\bar e_{n}, \\
&& \bar \varepsilon_{BI_a(n)}=\sum_{i=1}^n \bar e_i, \\
&& \bar \varepsilon_{BI_b(n)}=\sum_{i=1}^p \bar e_i, \\
&& \bar \varepsilon_{BII(n)}=\bar e_1, \\
&& \bar \varepsilon_{CI(n)}=\sum_{i=1}^n \bar e_i, \\
&& \bar \varepsilon_{CII_a(2k)}=\sum_{i=1}^k \bar e_{2i}, \\
&& \bar \varepsilon_{CII_b(2k)}=\sum_{i=1}^p \bar e_{2i}, \\
&& \bar \varepsilon_{DI_a(n)}=\bar \varepsilon_{DI_b(n)}=\sum_{i=1}^n \bar
e_i, \\
&& \bar \varepsilon_{DI_c(n)}=\sum_{i=1}^p \bar e_i, \\
&& \bar \varepsilon_{DII(n)}=\bar e_1, \\
&& \bar \varepsilon_{DIII_a(2k+1)}=\bar e_{2k+1}+\sum_{i=1}^k \bar e_{2i},\\
&& \bar \varepsilon_{DIII_b(2k)}=\sum_{i=1}^k \bar e_{2i}, \\
&& \bar \varepsilon_{G(2)}=\bar e_1+\bar e_2,\\
&& \bar \varepsilon_{FI(4)}=\bar e_1+\bar e_2+\bar e_3+\bar e_4, \\
&& \bar \varepsilon_{FII(4)}=\bar e_4,\\
&& \bar \varepsilon_{EI(6)}=\bar \varepsilon_{EII(6)}=\bar e_1+\bar e_2+\bar
e_3+\bar e_4+\bar e_5+\bar e_6, \\
&& \bar \varepsilon_{EIII(6)}=\bar e_1+\bar e_5+\bar e_6, \\
&& \bar \varepsilon_{EIV(6)}=\bar e_1+\bar e_5,\\
&& \bar \varepsilon_{EV(7)}=\bar e_1+\bar e_2+\bar e_3+\bar e_4+\bar
e_5+\bar e_6+\bar e_7, \\
&& \bar \varepsilon_{EVI(7)}=\bar e_1+\bar e_2+\bar e_3+\bar e_5, \\
&& \bar \varepsilon_{EVII(7)}=\bar e_1+\bar e_5+\bar e_6, \\
&& \bar \varepsilon_{EVIII(8)}=\bar e_1+\bar e_2+\bar e_3+\bar e_4+\bar
e_5+\bar e_6+\bar e_7+\bar e_8, \\
&& \bar \varepsilon_{EIX(8)}=\bar e_1+\bar e_5+\bar e_6+\bar e_7.
\end{eqnarray}


\end{appendix}

\end{document}